\documentclass[journal]{IEEEtran}

\ifCLASSINFOpdf
\else
\fi

\usepackage{cite}
\usepackage{amsmath, amsthm, amssymb, amsfonts, mathtools, eurosym, scalerel}
\usepackage{graphicx, float, subfigure, epsfig}

\newtheorem{thm}{Theorem}
\newtheorem{lemm}{Lemma}
\newtheorem{corollary}{Corollary}

\usepackage{stfloats}
\usepackage{booktabs, tabularx, makecell}

\begin{document}

\title{Stochastic Geometry Modeling and Analysis for THz-mmWave Hybrid IoT Networks}

\author{Chao~Wang,
        Young Jin Chun,~\IEEEmembership{Member,~IEEE}
\thanks{This work was supported by in part by the City University of Hong Kong (CityU), Startup Grant 7200618, and in part by the CityU, Strategic Research Grant 21219520}        
\thanks{C. Wang and Y. J. Chun are with the Department of Electrical Engineering, City University of Hong Kong, Hong Kong, China (e-mail: cwang224-c@my.cityu.edu.hk; yjchun@cityu.edu.hk)}}


\maketitle

\begin{abstract}
Terahertz (THz) band (0.1–10 THz) contains abundant spectrum resources that can offer ultra-high data rates. However, due to the THz band's inherent characteristics, i.e., low penetrability, high path loss, and non-negligible molecular absorption effect, THz communication can only provide limited coverage. To overcome these fundamental obstacles and fully utilize the THz band, we consider a hybrid Internet-of-Things (IoT) network consisting of THz and millimeter wave (mmWave) cells. A hybrid IoT network can dynamically switch between mmWave and THz links to ensure reliable and ultra-fast data connection. We use a stochastic geometric framework to evaluate the proposed hybrid IoT network's coverage probability and spectral efficiency and validate the analysis through numerical simulation. In this paper, we derive a closed-form expression of the Laplace transform of the interference while considering an accurate multi-level ﬂat-top (MLFT) antenna pattern. We observed that a large antenna array with a strong bias to the THz base station (TBS) improves the end-to-end network performance through numerical results. Furthermore, we recognized a fundamental trade-off relation between the TBS's node density and the bias to mmWave/THz; \textit{e.g.}, high TBS density with a strong bias to the TBS may degrade the network performance. 
\end{abstract}


\IEEEpeerreviewmaketitle

\section{Introduction}
The emergence of a tremendous amount of rate-greedy applications and a rapidly increasing number of mobile user equipment (UE) and IoT devices, together with a growing demand for faster data rates, are harshly depleting the spectrum resources. With its ample spectrum resource, the THz band is envisioned as the key enabler of the ultra-high data rate of Terabits-per-second (Tbps). However, despite its potential in the next generation network system, there are some obstacles to overcome to take full advantage of the THz band. 

The propagation characteristics of the THz band are affected by several factors. First, THz waves undergo severe path loss due to their ultra-high frequency, leading to a limited transmission range and low penetrability of THz communication. Second, the molecular absorption loss caused by converting signal energy into molecules' kinetic energy degrades the THz transmission. It has been revealed that the absorption loss caused by water vapor and oxygen molecules has a substantial impact on the THz communication \cite{jornet2011channel}. Third, due to the power constraint, directional antennas are utilized, generating highly directive THz transmission. Given the above, THz communication has the characteristics of high data rate, strong directivity but low penetrability, and limited coverage. A promising solution to overcome these drawbacks is to densely deploy many THz small cells equipped with highly directional antennas.  

Despite the very limited THz communication coverage, the THz band's IoT networks are envisioned as promising applications. Nevertheless, the stand-alone deployment of TBS is not sufficient to achieve universal coverage. A compromised but reliable strategy is to deploy TBS over ongoing massively deployed mmWave networks in a hybrid manner so that UEs can access either TBSs or the mmWave base stations (MBSs) depending on their link quality and association strategy. This hybrid IoT network composed of THz and mmWave cells is expected to support both high data rate and extended coverage. In specific, a UE can choose a THz link to support its service when the THz link outperforms the mmWave links according to the association strategy. Otherwise, when the THz link quality drops below the required threshold, the handoff event occurs, and the UE will connect to MBSs.

In this study, we want to model and analyze the hybrid IoT networks consisting of TBSs and MBSs using stochastic geometry methods. We first analyze the signal-to-interference-plus-noise ratio ($\tt SINR$) distribution of THz-only and mmWave-only networks, respectively. The coverage probability and spectral efficiency will be investigated based on the $\tt SINR$ distribution. With the aforementioned expressions in hand, the performance of the hybrid IoT network will be evaluated. The theoretical accuracy will be validated with simulations.

The rest of this paper is organized as follows. In Sec II, we will review the previous works about THz networks and hybrid IoT networks. In Sec III, the network topology and system model will be elaborated. The association probability and conditioned distance distribution will also be derived. In Sec IV, we focus on the theoretical framework. The expressions of coverage probability and spectral efficiency will be provided. In Sec V, we present the numerical results and discussion about the impact of various parameters on the network performance, and we conclude the paper in Sec VI.

\section{Related works}
The THz band is envisioned as the key enabler of the next-generation wireless communication system. A variety of previous works have investigated the THz propagation characteristics and THz network performance. In \cite{jornet2011channel}, the THz channel model was comprehensively studied. This work numerically investigated the THz propagation characteristics, i.e., the molecular absorption loss, severe path loss, and molecular absorption noise. The mathematical framework of the THz channel model was proposed, and the channel capacity was evaluated for different medium compositions. The results revealed the significant impact of molecular absorption by water vapor molecules, which demonstrated the necessity of molecular absorption loss in THz channel modeling.

Stochastic geometry models are widely utilized to evaluate the performance of THz cellular networks and THz-supported IoT networks. In \cite{petrov2017interference}, the THz networks were modeled by the Poisson point process (PPP), and the aggregated network interference and $\tt SINR$ were evaluated. The proposed model considered the peculiar THz channel characteristics, e.g., the molecular absorption loss, severe path loss of THz, molecular absorption noise, and the Johnson-Nyquist noise, as well as simplified directional antenna pattern and blockage probability. This work built a mathematical framework to analyze the performance of THz networks. Based on this model, the impact of the different properties of the THz channel on the network's performance was revealed in detail. In \cite{kokkoniemi2017stochastic}, the authors focused on the interference power and outage probability of THz networks, which adopted the channel model proposed in \cite{jornet2011channel}. The same THz channel model was also adopted in \cite{liu2020covert} which investigated the convert communication in IoT networks in the THz band, and the mean and variance of the interference in THz networks were derived. It was revealed in \cite{liu2020covert} that THz links are of more security than AWGN channels.

However, an obstacle in analyzing $\tt SINR$ of THz networks is deriving the closed-form expression of interference distribution. Most of the current works adopted asymptotic models to approximate the $\tt SINR$. In \cite{petrov2017interference}, Taylor expansion was utilized to estimate the $\tt SINR$. In \cite{kokkoniemi2017stochastic}, the mean interference power and outage probability were investigated by approximating the interference as log-logistic distribution. Unlike previous works, the wireless local area networks in the THz band (T-WLAN) were investigated in \cite{wu2019interference} where the APs and UEs were confined in an indoor environment. This work modeled the small-scale THz fading phenomena with the Nakagami-$m$ fading model and utilized log-normal distribution to approximate interference distribution. In \cite{hossain2019stochastic}, the authors utilized the TS-OOK modulation scheme \cite{6804405} and stochastic geometry to evaluate the multi-user interference of THz networks and experimentally validated.

Considering the current wireless communication infrastructure, it is highly possible that the THz networks will be overlaid on mmWave networks. Additionally, this hybrid deployment manner can also complement the THz networks to optimize the overall system coverage. This two-tier hybrid IoT network with the coexistence of THz and mmWave band is essentially a type of heterogeneous network (HetNet) while there is no inter-tier interference with each other. Similar researches have been conducted in the case where mmWave and LTE coexist. In \cite{elshaer2016downlink}, a hybrid cellular network with the coexistence of mmWave and sub-6 GHz was investigated. In this work, both the mmWave and sub-6 GHz tiers were modeled as PPP networks. With consideration of decoupled association, a UE will independently connect to a mmWave base station (BS) or a sub-6 GHz BS on both the uplink and downlink. Association strategies and coverage analysis for $\tt SINR$ and rate were numerically analyzed. In \cite{andrews2016modeling}, the authors analyzed the hybrid network from a general perspective. The rate coverage expression of the hybrid network was proposed. Furthermore, it was also pointed out that this type of HetNet can also be modeled with clustered point processes such as the Neyman-Scott process, where the sub-6 GHz BSs will be located at the center of clusters while the mmWave BSs will be deployed around the central sub-6 GHz BSs. In \cite{jo2012heterogeneous}, the heterogeneous cellular networks consisting of multi-tier networks with flexible cell association sharing the same spectral band were investigated. Different from the proposed THz-mmWave hybrid IoT networks, inter-tier interference has to be considered in this heterogeneous network. A similar hybrid network with the coexistence of THz and radio frequency (RF) networks was proposed in \cite{sayehvand2020interference}. This study considered the THz propagation characteristics and simplified directional antenna patterns. A closed-form expression of aggregate interference was derived to evaluate the coverage probability of THz networks. In this study, an opportunistic RF/THz network and a hybrid network were investigated. In the opportunistic RF/THz network, a user can associate to the BS with maximum biased received signal power (BRSP), whereas, in the hybrid network, a user can associate with both nearest RF and THz BS.

Another crucial factor in mmWave and THz networks is the directional antenna. Most of the previous works \cite{wu2019interference, petrov2017interference,sayehvand2020interference} adopted the binarized antenna pattern, i.e., the flat-top pattern, which only contain the main lobe and side lobes in order to simplify the derivation. It is demonstrated in \cite{yu2017coverage} that this flat-top antenna pattern can only reveal limited insights and lead to significant deviation from the performance of the actual antenna pattern. However, the actual antenna pattern is fairly intractable in the analysis. To improve analytical tractability, two approximate antenna patterns were proposed in \cite{yu2017coverage} and achieved decent agreement of the performance of the actual antenna pattern. In \cite{wei2018simple}, a discrete approximation antenna pattern was proposed to make the derivation more tractable and maintained a decent accuracy.

To summarize, evaluating the performance of the THz network is still an open problem. Furthermore, there is so far no investigation on the THz-mmWave hybrid IoT networks, which is a promising trend to optimize the utilization of the THz band. In this work, we want to reveal more insights into THz networks and the performance of the proposed hybrid IoT networks, utilizing accurate directional antenna patterns. The closed-form expression of the Laplace transform of interference in THz and mmWave networks will be derived. Furthermore, the coverage probability and spectral efficiency will be analyzed to evaluate the performance of the network.

\section{System Model}
\subsection{Network Topology}
We consider a downlink (DL) of a two-tier hybrid IoT network formed by MBSs and TBSs. Each tier is modeled as independent homogeneous Poisson point process (PPP) $\Phi_M$ with density $\lambda_m$ for MBSs (or $\Phi_T$ with density $\lambda_T$ for TBSs) on a $\mathbb{R}^2$ plane as illustrated in Fig. \ref{system model}. The UE nodes are distributed by a PPP $\Phi_u$ with density $\lambda_u$. Without loss of generality, we assume that the typical user is located at the origin. We represent the link between the typical UE and the $i$-th BS as $x_{i}$, where we use subscript $m$ and $T$ to differentiate between mmWave and THz link, \textit{i.e.}, $x_{m, i}$ is the $i$-th mmWave link and $x_{T, i}$ is the $i$-th THz link. The intended link between the typical UE and the associated BS is assumed $i = 0$ and we removed the subscript to simplify the notation, \textit{i.e.}, $x_{m, 0} = x_{m}$, $x_{T, 0} = x_{T}$. Each UE is equipped with an omnidirectional antenna with antenna gain $G_u = 1$. Additionally, the BSs allocate a radio resource block to the associated UEs in an orthogonal way. The parameters used in this work are summarized in Table I.

\begin{figure*}[!htp]
  \centering
  \subfigure[]
  {\includegraphics[width=0.43\linewidth]{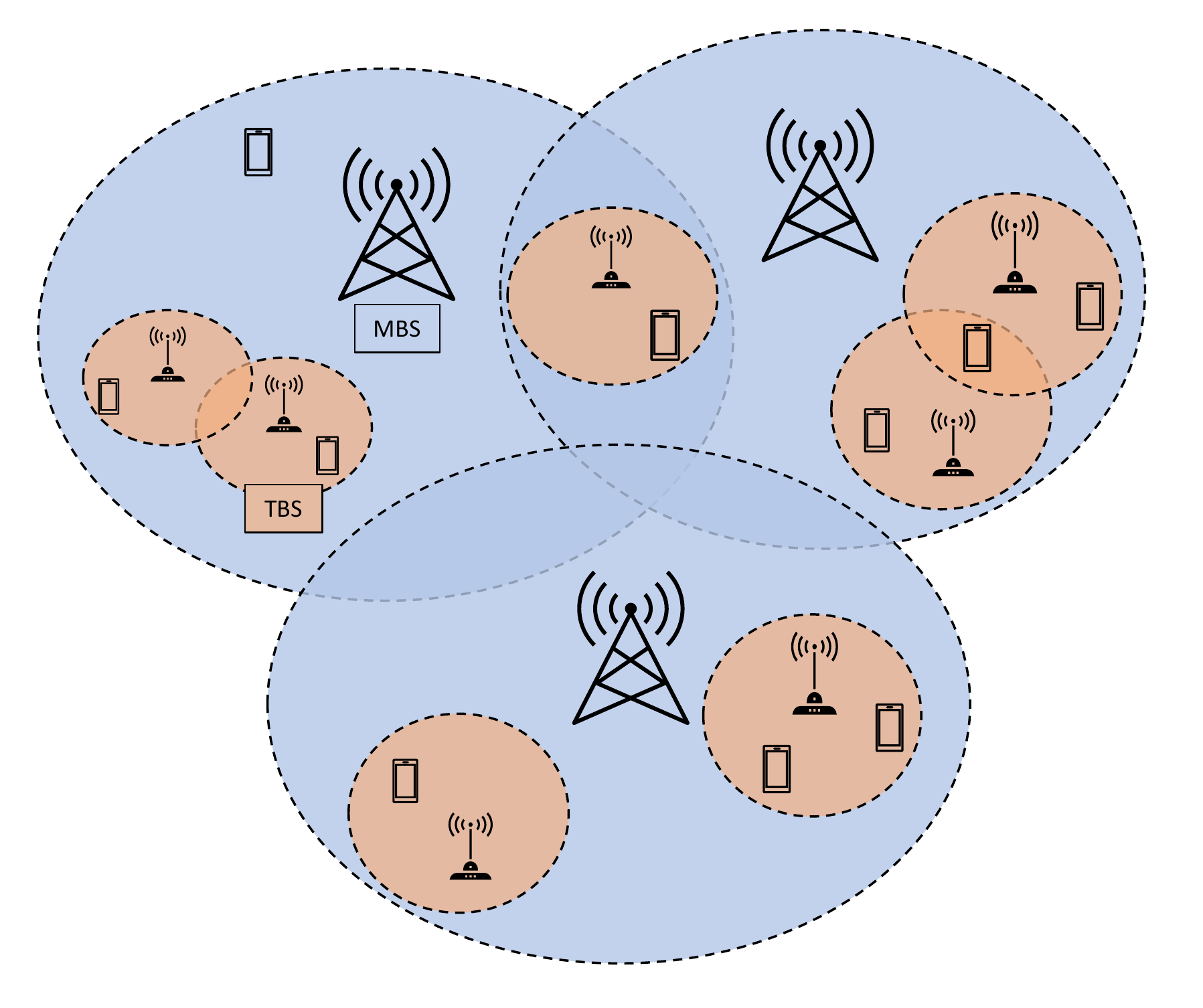}
  \label{system model}
  }\quad
  \subfigure[]
  {\includegraphics[width=0.43\linewidth]{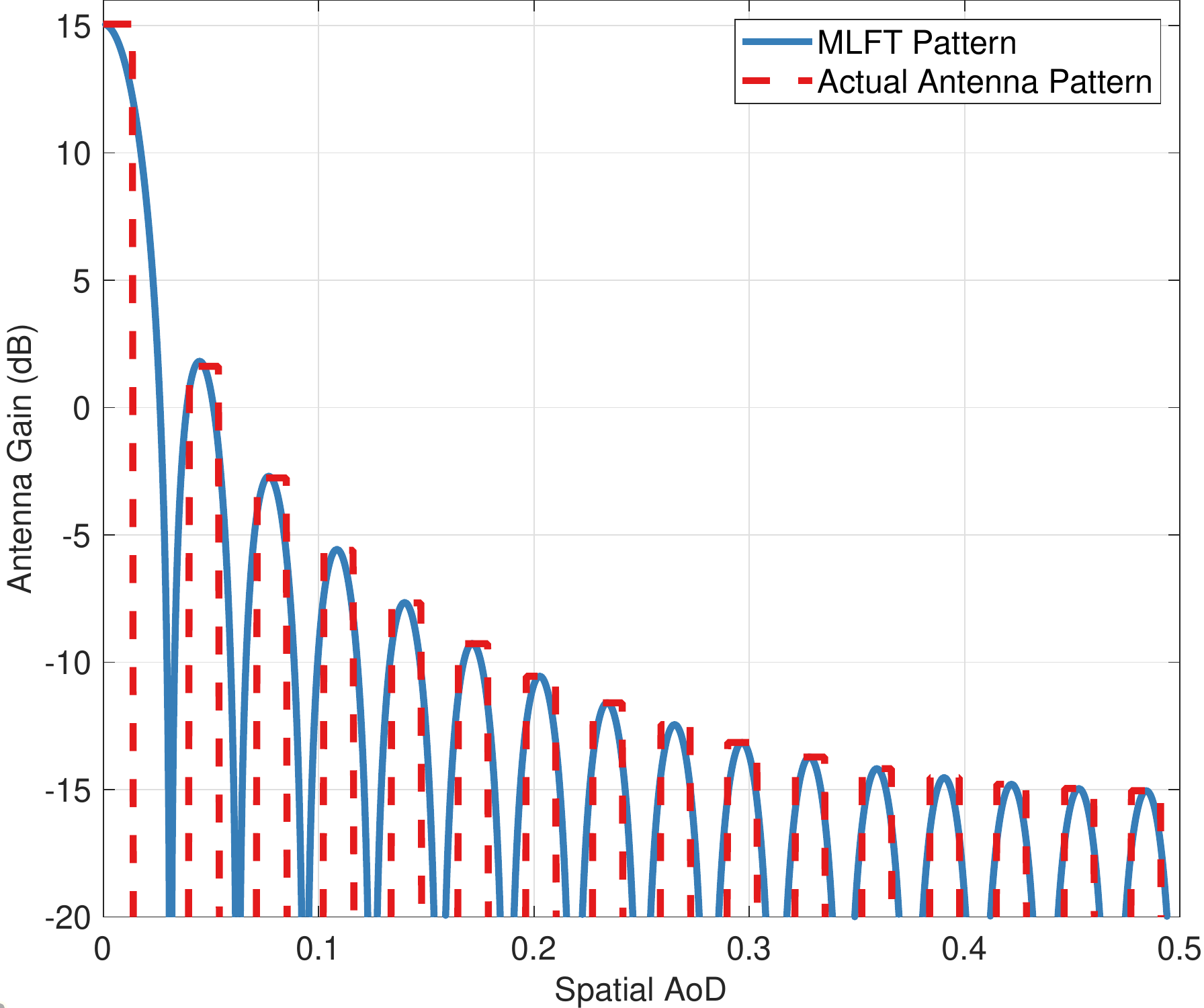}
  \label{antenna pattern}
  }
\caption{(a) System model of the hybrid IoT network consisting of TBSs and MBSs, which are modeled as two independent PPPs, (b) Visualization of the adopted MLFT and actual antenna pattern when $N_t$ = 32.}
\label{system model & antenna pattern}
\end{figure*}

\begin{table}[!htp]
\centering
\caption{NOTATION AND SIMULATION PARAMETERS}\label{tab:parameter}
\begin{tabularx}{0.48\textwidth}{ 
  | >{\raggedright\arraybackslash}c 
  | >{\centering\arraybackslash}c 
  | >{\centering\arraybackslash}X | }
\hline
\textbf{Notation} & \textbf{Parameter} & \textbf{Value(if applicable)}\\
\hline
$ \Phi_T,\lambda_T$ & TBS PPP and Density & $0.05/m^2$\\
\hline
$\Phi_m,\lambda_m$ & MBS PPP and Density & $5\times10^{-4}/m^2$\\
\hline
$x_T,x_m$ & Closest link distance to TBS and MBS &  \\
\hline
$f_T,f_m$ & THz and mmWave acrrier frequency & 350 GHz, 30 GHz\\
\hline
$P_T, P_m$& Transmitting power of TBS and MBS& 73 dBm, 53 dBm \\
\hline
$\sigma^2$ & Noise power for mmWave networks & -85 dBm\\
\hline
$B_T, B_m$& Bias to TBS and MBS& 10, 1 \\
\hline
$\alpha_T, \alpha_m$& Pathloss exponent for TBS and MBS& 4, 2 \\
\hline
$R_T, R_m$& Radius of coverage area for TBS and MBS& 100m, 20m \\
\hline
$G_T,G_m$ & Antenna gain of TBS and MBS &  \\
\hline
$N_T, N_m$& Antenna array size of TBS and MBS& 64, 16 \\
\hline
$M$& Nakagami-$m$ fading parameter& 4 \\
\hline
$T$& Temperature in Kelvin & 296 K \\
\hline
$\xi$& Relative humidity & 0.6 \\
\hline
$p$& Air pressure & 101325 Pa \\
\hline

\end{tabularx}
\end{table}

\subsection{Directional Antenna Pattern} 
For both mmWave and THz links, highly directional antennas are indispensable to combat severe propagation loss. The commonly adopted antenna pattern model for theoretical analysis is the simplified flat-top antenna model, which binarized the antenna gain into two levels to improve tractability but at the cost of a low approximation accuracy and failure to reveal the diverse impact of the directional antenna pattern \cite{yu2017coverage}. To resolve this issue, we utilized the multi-level flat-top (MLFT) antenna pattern model, which was initially proposed in \cite{wei2018simple} and demonstrated to achieve a close approximation to the actual antenna pattern (illustrated in Fig. \ref{antenna pattern}) and coverage probability, while maintaining similar tractability level as the simplified flat-top model. The antenna gain function $G\left(\varphi\right)$ of the MLFT model in this paper is given by 
\begin{equation}
  G\left(\varphi\right) = 
  \begin{dcases}
    &G_k, \quad \text{if } \varphi_k - \frac{\psi}{2} \leq |\varphi| < \varphi_k + \frac{\psi}{2},\\
    &0, \quad \text{otherwise},
  \end{dcases}
  \label{eqn:MLFT} 
  \end{equation}
where $N_t$ is the number of antenna elements, $K = \bigl \lfloor \frac {N_t}{2}\bigr \rfloor$, $1 \leq k \leq K$, $\varphi$ is the cosine direction corresponding to the angle of departure (AoD), $\varphi_1 = \frac{\psi}{2}$, $G_1 = N_t$, $\varphi_k = \frac{2 k - 1}{2 N_t}$ and $G_{k} = G_{\mathrm {act}}(\varphi _{k})$ for $k \geq 2$, and 
$\psi$ is the half-power beamwidth (HPBW) defined by the actual antenna pattern $G_{\text{act}}\left(\varphi\right)$ as follows: 
\begin{equation} 
  G_{\text{act}}\left(\varphi\right) = 
  \frac {\sin^{2}(\pi N_t \varphi)}{N_t \sin^{2}(\pi \varphi)}, \quad G_{\text{act}}\left(\psi\right) = \frac{N_t}{2}.
  \label{eqn:actual_ULA}
  \end{equation}
We assumed that all TBSs and MBSs adopted the MLFT antenna pattern with $N_T$ and $N_m$ elements, respectively, and the BSs serve the associated UE with perfect beam alignment.

\subsection{Propagation Model}
For a mmWave link, the received signal power at the typical UE from an MBS at location $x_m$ can be expressed as $P_m G_m(\phi_{m})hL_m(x_{m})$, where $P_m$ is the transmitting power, $G_m(\phi_m)$ is the antenna gain of MBS with AoD $\phi_m$, $h$ denotes the small scale fading. The pathloss term $L_m(x_{m})= \left({c}/{4\pi f_m}\right)^2 x_{m}^{-\alpha_m}$ is composed of the pathloss exponent $\alpha_m$, the speed of light $c$, and operating frequency $f_m$. We assume Rayleigh fading to model the small scale fading over mmWave link, which is a widely adopted assumption \cite{elshaer2016downlink}. 

In contrast, THz propagation undergoes severe attenuation, including propagation loss, absorption loss, molecular absorption noise, and Johnson-Nyquist noise \cite{jornet2011channel, petrov2017interference, wu2019interference}. We denote the propagation loss over the THz link as $L_{P}(x_T)$, the molecular absorption loss as $L_{A}(x_T)$, the absorption noise as $N_{A}$, and the Johnson-Nyquist noise as $N_{JN}$, where
\begin{equation}
    \begin{split}
        L_{P}(x_T) &= \left(\frac{c}{4\pi f_T}\right)^2 x_T^{-\alpha_T}, \quad L_{A}(x_T) = \mathrm{e}^{-k_a(f_T) x_T},\\
        N_{A} &= P_T G_T L_{P}(x_T) \left(1 - L_{A}(x_T)\right),\\
        N_{JN} &= \frac{\hbar f_T}{{exp}\left(\hbar f_T /k_B T\right)-1},
    \end{split}
    \label{eqn:thz-secIII-c-1}
\end{equation}
$f_T$ is the operating frequency of the THz link, $\alpha_T$ is the pathloss exponent, $\hbar$ denotes the Planck's constant, $k_{B}$ is the Boltzmann constant and $T$ represents the temperature in Kelvin. The Johnson-Nyquist noise remains constant until 0.1 THz, then declines non-linearly. The molecular absorption loss $L_{A}(x_T)$ can be obtained by referring the HITRAN database \cite{rothman2014hitran} or using a simplified absorption model proposed in \cite{kokkoniemi2018simplified} that works for frequencies in the $275$-$400$ GHz band. In this work, we adopted the simplified adsorption model \cite{kokkoniemi2018simplified}. 

The received signal power at the typical UE from a TBS at location $x_T$ is evaluated as $S = P_T G_T(\phi_T) L_A\left(x_T\right) L_P\left(x_T\right) g$, where $G_T(\phi_T)$ is the antenna gain with AoD $\phi_T$, $g$ denotes the small scale fading. In \cite{wu2019interference} and \cite{wu2020interference}, Nakagami-$m$ fading has been demonstrated as a tentative model to approximate the small scale fading of THz link. We follow this approach and  model the power fading coefficient $g$ as a gamma distribution $g \sim \text{Gamma}\left(M, 1/M\right)$ with parameter $M$. 

Therefore, the received $\tt SINR$ at the typical user in mmWave links can be expressed as follows
\begin{equation}
    {\tt SINR}_m {=} \frac{P_m G_{m, 1} h_0 L_m(x_{m})}{\sigma_m^2 + I_m} 
    {=} \frac{P_m N_m h_0 \left(\frac{c}{4\pi f_m}\right)^2 {x_{m}}^{-\alpha_m}}{\sigma_m^2 + I_m}, 
    \label{eqn:SINR_mmWave}
\end{equation}
where $x_m$ is the distance of desired mmWave link, $\sigma_m^2$ refers to the noise power in mmWave link, $I_m = \sum_{i\in \Phi_m / o}P_mG_m(\phi_{m,i})h_i ({c}/{4\pi f_m})^2 x_{m,i}^{-\alpha_m}$ is the interference, and the second equality follows by the assumption of perfect beam alignment. Similarly, the received $\tt SINR$ at the typical user in THz links can be expressed as follows
\begin{equation}
    {\tt SINR}_T = \frac{S}{\sigma_T^2 + I_T} = \frac{P_T N_T e^{-k_a(f_T)x_T}g_0\left(\frac{c}{4\pi f_T}\right)^2 x_T^{-\alpha_T}}{\sigma_T^2 + I_T},
    \label{eqn:SINR_THz}
\end{equation}
where $S$ is the received signal power, $\sigma_T^2$ is the noise power in THz link, and the second equality follows by the assumption of perfect beam alignment. $I_T$ is the aggregated interference, which can be expressed as below.
\begin{equation}
\begin{split}
    S &= P_T N_T e^{-k_a(f_T)x_T}\left(\frac{c}{4\pi f_T}\right)^2 g_0 x_T^{-\alpha_T}\\
    \sigma_T^2 &= N_{JN} + \sum_{i\in \Phi_{T \backslash o}} G_T(\phi_{T,i}) P_T \left(1 - L_A\left(x_{T,i}\right) \right)L_P\left(x_{T,i}\right) g_i\\
    I_T &= \sum_{i\in \Phi_{T \backslash o}} G_T(\phi_{T,i})~ P_T L_A\left(x_{T,i}\right) L_P\left(x_{T,i}\right) g_i.
\end{split}    
\end{equation}

\subsection{Absorption Loss Model}
A unique feature of THz communication is the molecular absorption loss. To model the absorption normally will require many parameters which can be obtained from the spectroscopic databases, such as HITRAN database \cite{rothman2014hitran}. To intuitively reveal the impact of absorption on THz networks, we will adopt a simplified absorption model proposed in \cite{kokkoniemi2018simplified}, which has presented a high accuracy within the 275 – 400 GHz frequency band, which is also regarded as the potential frequency range for future THz communication system. According to \cite{slocum2013atmospheric}, this frequency range contains a relatively wide band with a lower absorption coefficient comparing to the higher frequency band. Given the operating frequency $f$ within 275 – 400 GHz, the absorption model is given as
\begin{equation}
    K_a(f) = y_1(f, \mu_{H_2 O}) + y_2(f, \mu_{H_2 O}) + \omega(f),
\end{equation}
where the $y_1(f, \mu_{H_2 O})$, $y_2(f, \mu_{H_2 O})$ and $\omega(f)$ are defined as
\begin{equation}
\begin{split}
    & y_1(f, \mu_{H_2 O}) = \\
    & \frac{0.2205\mu_{H_2 O}(0.1303\mu_{H_2 O} + 0.0294)}{(0.4093\mu_{H_2 O} + 0.0925)^2 + \left(\frac{f}{100c}-10.835\right)^2},
\end{split}    
\end{equation}
\begin{equation}
\begin{split}
    & y_2(f, \mu_{H_2 O}) = \\
    & \frac{2.014\mu_{H_2 O}(0.1702\mu_{H_2 O} + 0.0303)}{(0.537\mu_{H_2 O} + 0.0956)^2 + \left(\frac{f}{100c}-12.664\right)^2},
\end{split}    
\end{equation}
\begin{equation}
    \omega(f) = p_1 f^3 + p_2 f^2 + p_3 f + p_4,
\end{equation}
where $\mu_{H_2 O}$ is the volume mixing ratio of water vapor that is calculated based on relative humidity $\xi$ as follows
\begin{equation}
    \mu_{H_2 O} = \frac{\xi}{100}\frac{p_w^*(T,p)}{p}.
\end{equation}
$\frac{\xi p_w^*(T,p)}{100}$ is the partial pressure of water vapor and $p_w^*(T,p)$ is the saturated water vapor partial pressure for a given pressure $p$ and temperature $T$, which can be estimated by Buck equation \cite{alduchov1996improved} as 
\begin{equation}
    p_w^* = 6.1121(1.0007+3.46\times10^{-6}p)exp\left(\frac{17.502T}{240.94+T}\right).
\end{equation}
$\omega(f)$ is an equalization factor and the coefficients are given as: $p_1 = 5.54 \times 10^{-37}, ~p_2 = -3.94 \times 10^{-25}, ~p_3 = 9.06 \times 10^{-14}$ and $p_4 = -6.36\times10^{-3}$. It has been shown in \cite{kokkoniemi2018simplified} that this simplified absorption model achieved a fairly decent accuracy comparing to the exact absorption loss. In this work, we will utilized this model in standard  atmospheric conditions, i.e temperature of 298 $K$, pressure of 101325 Pa and relative humidity of 0.6. The absorption coefficient calculated by this simplified absorption model is illustrated in Fig. \ref{kaf}. As it shows, within the frequency band of 275 - 400 GHz, there are two peaks of absorption coefficient. According to \cite{slocum2013atmospheric}, the absorption coefficient within the adopted frequency band is relatively lower and more steady comparing to higher frequency band. When the operating frequency is within 500 - 600 GHz, the absorption coefficient will reach the peak of almost 10 $m^{-1}$ \cite{slocum2013atmospheric}. Later we will reveal how the absorption coefficient affects the performance of THz networks.

\begin{figure}[t] 
\centering 
\includegraphics[width=\linewidth]{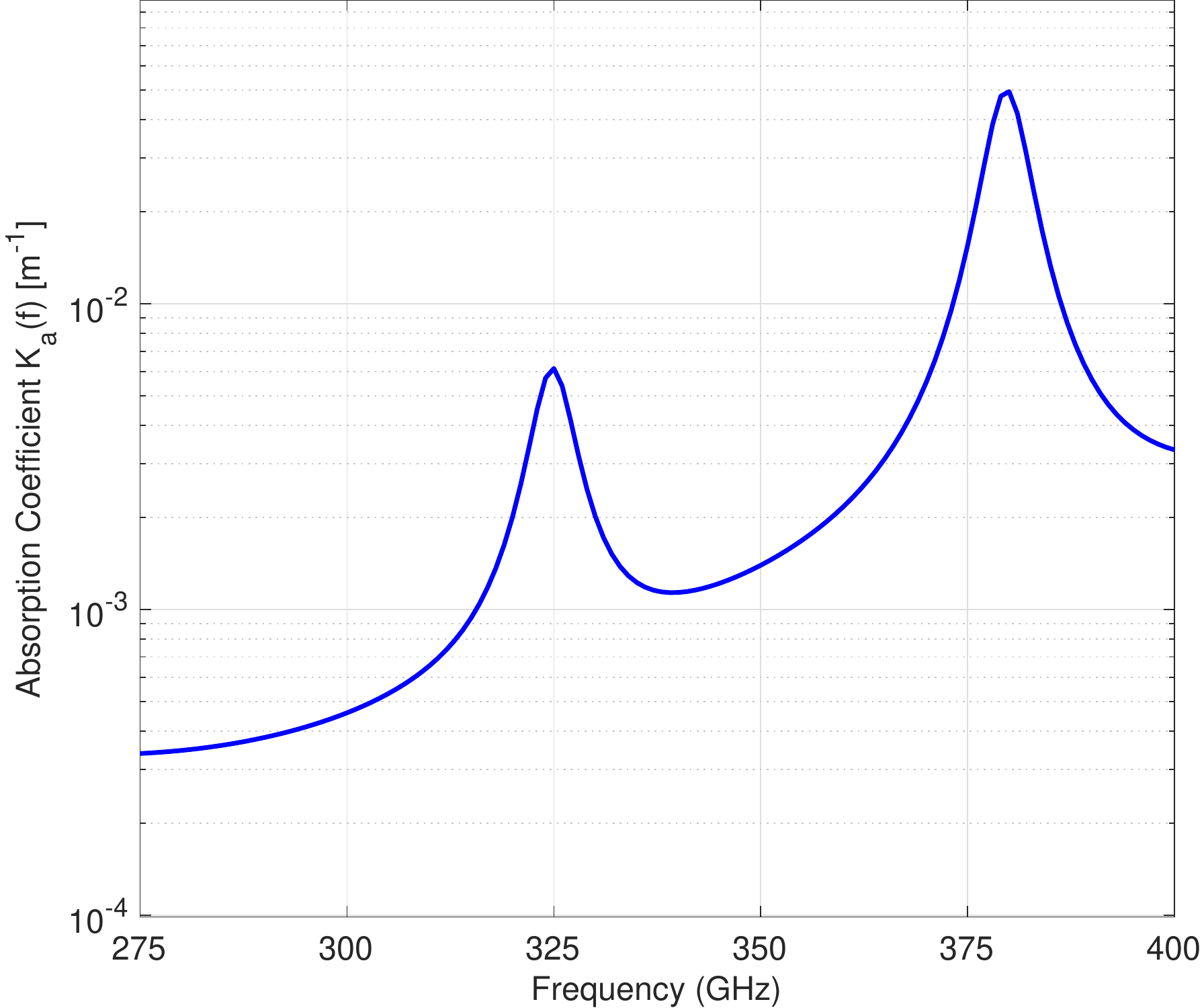}
\caption{Absorption coefficient $K_a(f)$ calculated by simplified absorption model within frequency band of 275 - 400 GHz.} 
\label{kaf} 
\end{figure}

\subsection{Blockage Model}
In this study, we adopt the equivalent LOS ball model proposed in \cite{bai2014coverage}, where all links within a ball of radius $R_B$ are regarded as LOS links, and the links beyond $R_B$ are regarded as none-line-of-sight (NLOS) links. We define $R_T$ and $R_m$ as the radius of LOS balls for THz and mmWave networks. Due to the poor penetrability of THz transmission and negligible impact of NLOS interference in mmWave networks \cite{bai2014coverage}, we assume all the NLOS links are blocked.

\subsection{Maximum Biased Power Association}
In this study, we considered a hybrid IoT network with the coexistence of TBSs and MBSs. Therefore, a UE can have access to either TBS or MBS according to the predefined association strategy, which we considered the Max-BRP association policy \cite{elshaer2016downlink}. This association strategy assumes a UE is associated with the strongest BS in terms of long-term averaged biased received power, and thus fading is averaged out. The intended biased received power from TBS and MBS can be expressed as
\begin{align}
    P_{r,T} &= N_T P_T B_T e^{-k_a(f_T)x_T}\left(\frac{c}{4\pi f_T}\right)^2 x_T^{-\alpha_T}\\
    P_{r,M} &= N_m P_m B_m \left(\frac{c}{4\pi f_m}\right)^2 x_m^{-\alpha_m},
\end{align}
where $B_T$ and $B_m$ are the positive biased factors for TBSs and MBSs tier, respectively. Specifically, when $\frac{B_T}{B_m} > 1$, UEs are more encouraged to associate with TBSs. On the contrary, UEs tend to associate with MBSs if $\frac{B_T}{B_m} < 1$. In particular, the Max-BRP turns into the conventional maximum received power association when $\frac{B_T}{B_m} = 1$.

\begin{thm}
\label{thm:associate_THz}
The probability of the typical UE is associated with THz tier is 
\begin{align}
\mathcal {A}_{\mathrm T} = \int_0^{R_T} f_{x_T}(x) e^{-\pi \lambda_m \left(\varepsilon x^{\alpha_T} e^{k_a(f_T)x}\right)^ {\frac{2}{\alpha_m}}}{\rm d}x,
\label{eqn:associate_THz}
\end{align}
where $\varepsilon = \frac{B_m P_m N_m}{B_T P_T N_T} \left(\frac{f_T}{f_m}\right)^2$ and $f_{x_T}(x) = \frac{2\pi\lambda_T}{1-e^{-\lambda_T \pi  R_T^2}}  x e^{-\pi\lambda_T x^2}$ is the distribution of the distance from typical UE to the nearest TBS. 
\end{thm}

\begin{proof}
Provided on the Max-BRP association strategy, the probability of the typical UE is associated with a TBS in the downlink can be calculated as follows:
\begin{align}
\begin{split}
\label{eqn:associate_prob_THz}
\mathcal {A}_{\mathrm T} &= \mathbb {E}_{x_T} \left[\mathbb {P} \left(P_{r,T} > P_{r,m}\right)\right] \\
&= \mathbb {E}_{x_T} \left[\mathbb {P} \left(x_m > \left(\varepsilon  x_T^{\alpha_T} e^{k_a(f_T)x_T}\right)^{\frac{1}{\alpha_m}}\right)\right]\\
&\overset{(a)}{=} \mathbb{E}_{x_T} \left[e^{-\pi \lambda_m \left(\varepsilon x_T^{\alpha_T} e^{k_a(f_T)x_T}\right)^ {\frac{2}{\alpha_m}}}\right]\\
&{=} \int_0^{R_T} f_{x_T}(x) e^{-\pi \lambda_m \left(\varepsilon x^{\alpha_T} e^{k_a(f_T)x}\right)^ {\frac{2}{\alpha_m}}}{\rm d}x,
\end{split}
\end{align}
where $(\ref{eqn:associate_prob_THz}a)$ follows the null probability of 2-D PPP with density of $\lambda$ in a circle with radius $\rho$ is $\mathbb {P} [r \geqslant \rho] = e^{-\pi \lambda \rho^2}$. Naturally, the association probability to MBSs is expressed as
\begin{equation}
   \mathcal {A}_{\mathrm m} = 1 - \mathcal {A}_{\mathrm T}.
   \label{eqn:associate_mmWave}
\end{equation}
Although Theorem \ref{thm:associate_THz} does not provide a closed-form expression, this integration can be efficiently calculated by mathematical software.
\end{proof}

\subsection{Conditioned Distance Distribution}
Denoting $X_k$ as the distance between the typical UE and the associated BS where $k\in \{T,m\}$ standing for THz and mmWave tier, we want to derive the probability density function (PDF) of $X_k$ in this subsection. Different from the distance from typical UE to the nearest BS in a PPP network, $f_{X_k}(\hat{x})$ can be interpreted as the distance from the associated BS to the typical UE conditioned on the typical UE is associated with THz/mmWave tier.
\begin{lemm}
\label{lemm:serving_BS_distance_distribution}
The PDF of distance from the typical UE to its associated BS when it is associated to the THz tier and mmWave tier are as given follows:
\begin{align}
\begin{split}
    &f_{X_T}(\hat{x})=\frac{1}{\mathcal {A}_{\mathrm T}} \left[f_{x_T}(\hat{x}) e^{-\pi \lambda_m \left(\varepsilon \hat{x}^{\alpha_T} e^{k_a(f_T)\hat{x}}\right)^ {\frac{2}{\alpha_m}}}\right]\\
    &f_{X_m}(\hat{x})=\frac{1}{\mathcal {A}_{\mathrm m}} \left[f_{x_m}(\hat{x})e^{-\pi\lambda_T \nu^2(\hat{x})}\right],
\end{split}    
\end{align}
where $\varepsilon$ and $f_{x_T}(\hat{x})$ are defined in Theorem \ref{thm:associate_THz}, $f_{x_m}(\hat{x})$ is the distribution of distance between the typical UE and its closest MBS, $\nu(\hat{x})=\frac{\alpha_T}{k_a(f_T)}W\left(\frac{k_a(f_T)}{\alpha_T}\left[\frac{\hat{x}^{\alpha_m}}{\varepsilon}\right]^{\frac{1}{\alpha_T}}\right)$, $W(z)$ is the Lambert $W$ function, and $\mathcal {A}_{\mathrm T}$ and $\mathcal {A}_{\mathrm m}$ are the association probabilities to THz and mmWave tier derived in Theorem \ref{thm:associate_THz}.

\begin{proof}
See Appendix \ref{appendix:proof_distance_distribution}.
\end{proof}

\end{lemm}

\section{Coverage Analysis}
In this section, we will elaborate on the mathematical framework to analyze the system. We first analyze the $\tt SINR$ coverage of TBS-only networks and MBS-only networks, respectively. The closed-form expression of Laplace transforms of the interference of THz and mmWave networks will be derived. Furthermore, the coverage probability, which is defined as the probability of $\tt SINR$ achieving the desired threshold, will be provided to evaluate the network performance.

Given the per-tier association probability of the typical UE which is provided in Theorem \ref{thm:associate_THz}, the coverage probability of the hybrid IoT networks can be expressed as
\begin{equation}
    \mathcal P_C(\tau) = \mathcal {A}_{\mathrm T}  \mathcal P_{C_T}(\tau) + \mathcal {A}_{\mathrm m} \mathcal P_{C_m}(\tau),
    \label{eqn:network_coverage_probability}
\end{equation}
where $\tau$ is the desired threshold, $\mathcal {A}_{\mathrm T}$ and $\mathcal {A}_{\mathrm m}$ has been given in Theorem \ref{thm:associate_THz}, and $\mathcal P_{C_T}(\tau)$ and $\mathcal P_{C_m}(\tau)$ are the coverage probabilities conditioned on that typical UE is associated with THz networks and mmWave networks, respectively. Therefore, to evaluate the aforementioned conditioned coverage probability, we first analyze the coverage probabilities of TBSs and MBSs networks in stand alone deployment manner.

\subsection{Coverage Probability for THz Networks}
The coverage probability of THz networks is defined as
\begin{align}
\begin{split}
\mathbb {C}_T\left(\tau\right) = \mathbb {E} \left[ \mathbb {P}\left(\tt SINR \geq \tau\right) \right]
\end{split}
\end{align}
As mentioned before, a major obstacle in analyzing THz networks is the interference distribution analysis. Instead of utilizing different distributions to approximate the interference distribution, we utilize the Alzer's inequality \cite{alzer1997some} described below to approximate the coverage probability.
\begin{lemm}
\label{lemm:Alzer inequality}
(From \cite{alzer1997some}) Given $g$ is a normalized Gamma random variable with parameter $M$, the probability $\mathbb {P} (g < \gamma)$ is tightly upper bounded by
\begin{align}
    \mathbb {P} (g < \gamma) < (1 - e^{a\gamma})^M,
\end{align}
where $a = M(M!)^{-\frac{1}{M}}$.
\end{lemm}
According to Lemma \ref{lemm:Alzer inequality}, we can approximate the coverage probability for THz networks as
\begin{align}
\label{eqn:THz_coverage_1}
  \begin{split}
    & \mathbb {C}_T\left(\tau\right) = \mathbb {E} \left[ \mathbb {P}\left(\tt SINR_T \geq \tau\right) \right] \\
    & = \mathbb {E}_{\hat{J}, x_T} \left[ \mathbb {P}\left(g_0 \geq \frac{\tau \left(\hat{J} + \hat{N}\right)}{\hat{S}}\right) \right]\\
    &\overset{(a)}{\simeq} 1- \mathbb{E}_{\hat{J}, x_T}\left[ \left( 1- e^{-a P_1(x_T) \left(\hat{J} + \hat{N}\right)}\right)^M \right]\\
    &\overset{(b)}{=}\sum_{n=1}^{M} \binom{M}{n}\left( -1\right)^{n+1} \mathbb{E}_{\hat{J}, x_T}\left[e^{-P_n(x_T) \left(\hat{J} + \hat{N}\right)}
    \right],
  \end{split}
\end{align}
where $\tau$ is the $\tt SINR$ threshold, $g_0$ is the power fading coefficient of the desired link, and $\hat{S} = e^{-k_a(f_T) x_T} x_T^{-\alpha_T}~ g_0$, $\hat{J} = \sum_{i=1}^{N} x_{T,i}^{-\alpha_T} g_i \hat{G}_i$, $\hat{N} = N_{JN}\left(P_T N_T \left(c/4\pi f_T\right)^2\right)^{-1}$, and $\hat{G}_i = G_{T,i}/N_T$. We define $P_1(x_T) \triangleq \tau~ x_T^{\alpha_T} e^{k_a(f_T) x_T}$, $P_n(x_T) \triangleq a n P_1(x_T)$ and Alzer's inequality is applied in (\ref{eqn:THz_coverage_1}a) with $a = M\left(M! \right)^{-\frac{1}{M}}$, and (\ref{eqn:THz_coverage_1}b) follows the generalized binominal expansion. Therefore, $\mathbb {C}_T$ can be further derived as Theorem \ref{lemm:coverage_probability}.

\begin{thm}
\label{lemm:coverage_probability}
The coverage probability of TBS-only networks given the threshold $\tau$ is expressed as:
\begin{align}
    \mathbb {C}_T(\tau) {=} \sum_{n=1}^{M} \binom{M}{n}\left( -1\right)^{n+1} \int_{0}^{R_T} \mathcal{I}_T(x)  {\rm d}x,
    \label{eqn:coverage_probability}
\end{align}
where $M$ is the Nakagami-$m$ fading parameter and $\mathcal{I}(x)$ is derived in Appendix \ref{appendix:proof_lemma_coverage_probability} as $\mathcal{I}_T(x) = f_{x_T}(x) e^{-P_n(x) \hat{N}}\mathcal{L}_{\hat{J}}\left( P_n(x)\right)$, where $\mathcal{L}_{\hat{J}}\left( s \right)$ denotes the Laplace transform of $\hat{J}$.

\begin{proof}
See Appendix \ref{appendix:proof_lemma_coverage_probability}.
\end{proof}

\end{thm}

The integration of (\ref{THz_laplace_3}) yields the closed-form expression of the Laplace transform of the aggregated interference for TBS-only networks. Since the expression of the coverage probability contains only a single integration, the theoretical result can be efficiently evaluated by numerical integration.

\subsection{Coverage Probability for mmWave Networks}
Without loss of generality, we assume $h \sim exp(1)$. Combining the MLFT antenna pattern, the coverage probability for mmWave networks can be calculated similarly as below.

\begin{thm}
\label{lemm:mmwave_coverage_probability}
The coverage probability of MBS-only networks is expressed as:
\begin{align}
    \mathbb {C}_m(\tau) {=} \int_0^{R_m} f_{x_m}(x) e^{-\tau x^{\alpha_m} \sigma^2} \mathcal{L}_{\hat{I}_m}(\tau x^{\alpha_m}) {\rm d}x,
\end{align}
where $\tau$ is the $\tt SINR$ threshold for the MBS-only networks and $\sigma^2 = \frac{\sigma^2_m}{PmNmL_0}$, $f_{x_m}(x) = \frac{2\pi\lambda_m}{1-e^{-\lambda_m \pi R_m^2}}  x e^{-\pi\lambda_m x^2}$, $\hat{I}_m = \sum_{i\in \Phi_m / o}\hat{g_i}h_i x_{m,i}^{-\alpha_m}$, $\hat{g_i}=G_{m,i}/N_m$. \ref{appendix:proof_lemma_coverage_probability}.
\end{thm}

\begin{proof}
\begin{equation}
  \begin{split}
   &\mathbb {C}_m\left(\tau\right) =  \mathbb {E} \left[ \mathbb {P}\left(\tt SINR \geq \tau\right) \right]\\
   &= \mathbb {E}_{x_m,\hat{I}_m} \left[ \mathbb {P}\left(h_0 \geq \tau x_m^{\alpha_m} \left(\hat{I}_m + \sigma^2\right)\right) \right]\\
   &=\int_0^{R_m} f_{x_m}(x) e^{-\tau x^{\alpha_m} \sigma^2} \mathcal{L}_{\hat{I}_m}(\tau x^{\alpha_m}) {\rm d}x.
  \end{split}
  \label{eqn:mmwave_coverage_1}
\end{equation}
And $\mathcal{L}_{\hat{I}_m}(s)$ is the Laplace transform of $\hat{I}_m$. Given $x_m = x$, $\mathcal{L}_{\hat{I}_m}(s)$ can be calculated as below:
\begin{align}
   &\mathcal{L}_{\hat{I}_m}(s)\notag\\
   &= exp\left(-2\pi\lambda_m \int_x^{R_m} \left(1 - \mathbb{E}_{\hat{g}}\left[ \frac{1}{1+s\hat{g}t^{-\alpha_m}} \right]\right) t{\rm d}t\right)\notag\\
   &=exp\left(4\pi\lambda_m\psi_m \left[\chi_m(s)+ \frac{K_m}{2}(x^2{-}R_m^2)\right]\right),
  \label{eqn:mmwave_laplace_1}
\end{align}
where $\psi_m$ is the half-power beamwidth and $\chi_m(s)$ is expressed as:
\begin{align}
  \begin{split}
    &\chi_m(s) = \int_x^{R_m} \sum_{k=1}^{K_m}\frac{1}{(1 +s\hat{g_k}t^{-\alpha_m})} t{\rm d}t\\
    &=\sum_{k=1}^{K_m} \left[\frac{t^{\alpha_m+2}}{s\hat{g_k}(\alpha_m+2)}~{}_2F_1\left( 1,1+\frac{2}{\alpha_m};2+\frac{2}{\alpha_m};-\frac{t^{\alpha_m}}{s\hat{g_k}} \right) \right]_x^{R_m}.
  \end{split}
  \label{eqn:mmwave_laplace_2}
\end{align}
Similarly, $\mathbb {C}_m(\tau)$ can be easily calculated out by numerical integration. 
\end{proof}

Simplified results can be obtained for some plausible cases as below.

\begin{corollary}
When $\alpha_m=2$ and noise is neglected, the coverage probability $\mathbb {C}_m(\tau)$ can be expressed as
\begin{equation}
    \mathbb {C}_m = \int_0^{R_m} f_{x_m}(x) \mathcal{\hat{I}}_m(x)  {\rm d}x,
\end{equation}
\end{corollary}
where
\begin{equation}
    \mathcal{\hat{I}}_m(x) = \prod \limits_{k=1}^{K_m} \left(\frac{x^2+\tau x^2\hat{g_k}}{R_m^2+\tau x^2\hat{g_k}}\right)^{2\pi\lambda_m\psi_m\hat{g_k} \tau x^2}.
\end{equation}

\begin{proof}
 If $\alpha_m=2$, $\chi_m(x)$ can be expressed as
\begin{align}
  \begin{split}
    \chi_m(s) &= \int_x^{R_m} \sum_{k=1}^{K_m}\frac{1}{(1 +s\hat{g_k}t^{-2})} t{\rm d}t\\
    &=\sum_{k=1}^{K_m} \left[ \frac{t^2}{2} - \frac{s\hat{g_k}}{2} ln(t^2+s\hat{g_k}) \right]_x^{R_m}.
  \end{split}
  \label{eq:30}
\end{align}
We obtain $\mathcal{\hat{I}}_m(x)$ by plugging (\ref{eq:30}) into (\ref{eqn:mmwave_laplace_1}). 
\end{proof}

\subsection{Coverage Probability for hybrid IoT networks}
With association probability, the coverage probability in stand alone mode and conditioned distance distribution in hand, we can calculate the conditioned coverage probability $\mathcal P_{C_T}$ and $\mathcal P_{C_m}$ simply by substituting the distance distribution $f_{x_T}(x)$ and $f_{x_m}(x)$ in Theorem \ref{lemm:coverage_probability} and Theorem 
\ref{lemm:mmwave_coverage_probability} with the derived conditioned distance distribution $f_{X_T}(\hat{x})$ and $f_{X_m}(\hat{x})$ in Lemma \ref{lemm:serving_BS_distance_distribution}. Therefore, the coverage probability of the hybrid IoT network can be expressed by (\ref{eqn:network_coverage_probability}) as follows:
\begin{align}
    \mathcal P_C(\tau) &{=} \sum_{n=1}^{M} \binom{M}{n}\left( -1\right)^{n+1} \int_{0}^{R_T} {\cal Z}_T(\hat{x}) \mathcal{L}_{\hat{J}}\left( P_n(\hat{x})\right)  {\rm d}\hat{x} \notag\\
    &+ \int_0^{R_m} {\cal Z}_m(\hat{x}) \mathcal{L}_{\hat{I}_m}(\tau \hat{x}^{\alpha_m}) {\rm d}\hat{x},
\end{align}
where ${\cal Z}_T(\hat{x}) = \mathcal {A}_{\mathrm T}f_{X_T}(\hat{x}) e^{-P_n(\hat{x}) \hat{N}}$ and ${\cal Z}_m(\hat{x})= \mathcal {A}_{\mathrm m}f_{X_m}(\hat{x}) e^{-\tau \hat{x}^{\alpha_m} \sigma^2}$.

\section{Spectral efficiency}
In this section, we focus on the spectral efficiency of the hybrid IoT network, which can be calculated as below.
\begin{align}
\begin{split}
\label{eqn:hybrid_network_SE} 
    C &= \sum _{k \in \{T, m\}} \mathcal {A}_{k} C_{k} \\
    &= \sum _{k \in \{T, m\}} \mathcal {A}_{k} \mathbb {E}\left [{ \ln \left ({ 1 + \mathrm {SINR}_{k}(x) }\right) }\right ]
\end{split}
\end{align} 
where $\mathcal {A}_{k}$ is evaluated in Theorem \ref{thm:associate_THz}, $C_{k}$ is spectral efficiency conditioned on the typical user is associated with $k$-th tier. We will first derive the conditioned spectral efficiency of THz and mmWave networks, then the spectral efficiency of the hybrid IoT network will be evaluated by (\ref{eqn:hybrid_network_SE}).

\subsection{Spectral Efficiency of THz Networks}
We first derive $C_{T}$, the conditioned spectral efficiency of the THz tier. In \cite{andrews2011tractable,jo2012heterogeneous}, the spectral efficiency was calculated as a function of the coverage probability. Since in our work, the coverage probability of THz networks was approximated by Alzer's inequality. Therefore, to calculate the accurate expression of spectral efficiency, we utilize the method proposed in \cite{hamdi2007useful}.
\begin{lemm}
\label{lemm:useful technique}
(from \cite{hamdi2007useful}) Suppose $b$ is an constant, $g$ is a Nakagami-$m$ fading coefficient with parameter $M$ which follows Gamma distribution with unit mean, $y$ is an arbitrary random variable and $g$ and $y$ are independent, the explicit expression of $\mathbb{E}[ln\left(1 + \frac{g}{y+b}\right)]$ can be calculated as below.
\begin{align}
\begin{split}
    \mathbb{E}[ln\left(1 + \frac{g}{y+b}\right)]=\int\limits_{0}^{\infty}\zeta(z){\cal M} _{y}(Mz)e^{-Mzb}{\rm d}z,
\end{split}
\end{align} 
where ${\cal M}_{y}(z)={\rm E}[e^{-zy}]$ is the moment generating function of $y$, $\zeta(z)=\frac{1}{z}-\frac{1}{z(1+z)^M}$ which is derived in \cite[eq.(18)]{hamdi2007useful}.
\end{lemm}

Therefore, $C_{T}$ is evaluated as follows.
\begin{thm}
\label{lemm:THz_SE}
The spectral efficiency conditioned on the typical user associates with THz tier is
\begin{equation}
    \mathcal C_{T} = \int\limits_{0}^{R_T} \int\limits_{0}^{\infty}\zeta(z) \mathcal{L}_{\hat{J}}(\eta(\hat{x})) e^{-\eta(\hat{x}) \hat{N}} f_{X_T}(\hat{x}){\rm d}z{\rm d}\hat{x},
    \label{eq:34}
\end{equation}
where $\eta(\hat{x}) = Mz e^{ka(f_T)\hat{x}} \hat{x}^{\alpha_T}$, $\mathcal{L}_{\hat{J}}(s)$ and $\hat{N}$ are defined in Appendix \ref{appendix:proof_lemma_coverage_probability} and $f_{X_T}(\hat{x})$ is defined in Lemma \ref{lemm:serving_BS_distance_distribution}.
\end{thm}
\begin{proof}
Following Lemma \ref{lemm:useful technique} and (\ref{eqn:SINR_THz}), we can express $y = e^{ka(f_T)\hat{x}} \hat{x}^{\alpha_T} \hat{J}$ and $b = e^{ka(f_T)\hat{x}} \hat{x}^{\alpha_T} \hat{N}$, which obtains (\ref{eq:34}) readily. Although there is no closed-form solution for this expression, accurate spectral efficiency can be computed efficiently using mathematical software. 
\end{proof}

\subsection{Spectral Efficiency of mmWave Networks}
In this subsection, we will derive the spectral efficiency conditioned on the typical user is associated with the mmWave tier, which can be directly calculated as below:
\begin{align}
\begin{split}
    {C}_{m} &= \mathbb {E}_{X_m}\left [ \ln \left ({ 1 + {\tt SINR}_{m}(X_m) }\right)\right]\\
    &= \mathbb{E}_{X_m}\left [ \int\limits_{0}^{\infty}\mathbb {P}\left [ \ln \left ({ 1 + {\tt SINR}_{m}(X_m) }\right) > t \right]{\rm d}t \right]\\
    &= \mathbb{E}_{X_m} \left[ \int\limits_{0}^{\infty} \mathbb {P} [ h > (e^t-1)(\sigma^2 + \hat{I}_m) X_m^{\alpha_m} ] {\rm d}t  \right]\\
    &= \mathbb{E}_{X_m} \left[ \int\limits_{0}^{\infty} e^{(1-e^t)\sigma^2 X_m^{\alpha_m} } \mathcal{L}_{\hat{I}_m}((e^t-1)X_m^{\alpha_m} ) {\rm d}t \right]\\
    &= \int\limits_{0}^{R_m} \int\limits_{0}^{\infty} \frac{e^{-\tau\sigma^2 \hat{x}^{\alpha_m} }}{1+\tau} \mathcal{L}_{\hat{I}_m}(\tau \hat{x}^{\alpha_m}) f_{X_m}(\hat{x}) {\rm d}\tau{\rm d}\hat{x},
\end{split}
\label{eq:35}
\end{align}
where $\sigma$, $\hat{I}_m$ and $\mathcal{L}_{\hat{I}_m}(s)$ are all defined in Theorem \ref{lemm:mmwave_coverage_probability}, $f_{X_m}(\hat{x})$ is defined in Lemma \ref{lemm:serving_BS_distance_distribution} and the last step follows $\tau = e^t-1$.

\subsection{Spectral Efficiency for Hybrid IoT Networks}
The spectral efficiency of the hybrid IoT network can be evaluated by (\ref{eqn:hybrid_network_SE}) as follows
\begin{equation}
    \begin{split}
    &C = \int\limits_{0}^{R_T} \int\limits_{0}^{\infty} \mathcal{H}_T(z,\hat{x}) {\rm d}z{\rm d}\hat{x} + \int\limits_{0}^{R_m} \int\limits_{0}^{\infty} \mathcal{H}_m(\tau,\hat{x}) {\rm d}\tau{\rm d}\hat{x},\\
    &\mathcal{H}_T(z,\hat{x}) = \mathcal {A}_{T} \zeta(z) \mathcal{L}_{\hat{J}}(\eta(\hat{x})) e^{-\eta(\hat{x}) \hat{N}} f_{X_T}(\hat{x}),\\
    &\mathcal{H}_m(\tau,\hat{x}) = \mathcal {A}_{m} \frac{e^{-\tau\sigma^2 \hat{x}^{\alpha_m} }}{1+\tau} \mathcal{L}_{\hat{I}_m}(\tau \hat{x}^{\alpha_m}) f_{X_m}(\hat{x}),        
    \end{split}
\end{equation}
where the proof readily follows by using the law of total probability with (\ref{eq:34}) and (\ref{eq:35}).

\section{Numerical Results}
In this section, we will demonstrate the accuracy of the mathematical framework for the proposed hybrid IoT networks. The analytical results of the coverage probability and spectral efficiency will be illustrated and validated by simulation results which were averaged over $10^4$ realizations. 

\subsection{Coverage Analysis for THz Networks}
First, we will validate our analytical results of $\tt SINR$ coverage for TBS-only networks. In Fig. \ref{THz_pcov_Nt}, we illustrate the coverage probability of THz networks with antenna array size $N_t$ of 8, 16, 32, 64. The coverage probabilities were calculated by Theorem \ref{lemm:coverage_probability}. As it depicts, the analytical results closely match the simulation results. The results validate the tight bound of Lemma \ref{lemm:Alzer inequality} can provide an accurate approximation. The impact of antenna array size on THz network performance is also revealed in the results. We can observe that a larger antenna array size can significantly improve the coverage probability. This is due to the larger size of highly directional antennas leading to higher antenna gain, which remarkably compensates for the severe propagation loss resulting from the inherent high directivity and poor penetrability. Additionally, the narrower HPBW mitigates the interfering signals, which boosts $\tt SINR$. Fig. \ref{THz_pcov_ple} illustrates the coverage probability for varying pathloss exponent of 2, 3, 4 when antenna array size $N_t = 32$. Comparing to antenna array size, the pathloss exponent has a more significant impact on coverage probability. It is observed that a larger pathloss exponent will lead a better network performance. This is because a larger pathloss exponent will degrade the power of interfering links which, consequently, improves the coverage probability.

\subsection{Impact of Molecular Absorption}
Molecular absorption is a unique characteristic of THz communication. In this work, we adopted a simplified absorption model to calculate the absorption coefficient within the frequency band of 275 - 400 GHz. As illustrated in Fig. \ref{THz_pcov_kaf}, the absorption coefficient corresponding to varying operating frequency will lead to the fluctuations of coverage probability. In specific, the coverage probability is likely to be inversely proportional to the absorption coefficient. Intuitively, a higher absorption coefficient will cause severer absorption loss and, thus, lower coverage probability. Consequently, the absorption loss causes around 5\% degradation in terms of coverage probability within 275 - 400 GHz band.

\begin{figure*}[!htp]
  \centering
  \subfigure[]
  {\includegraphics[width=0.485\linewidth]{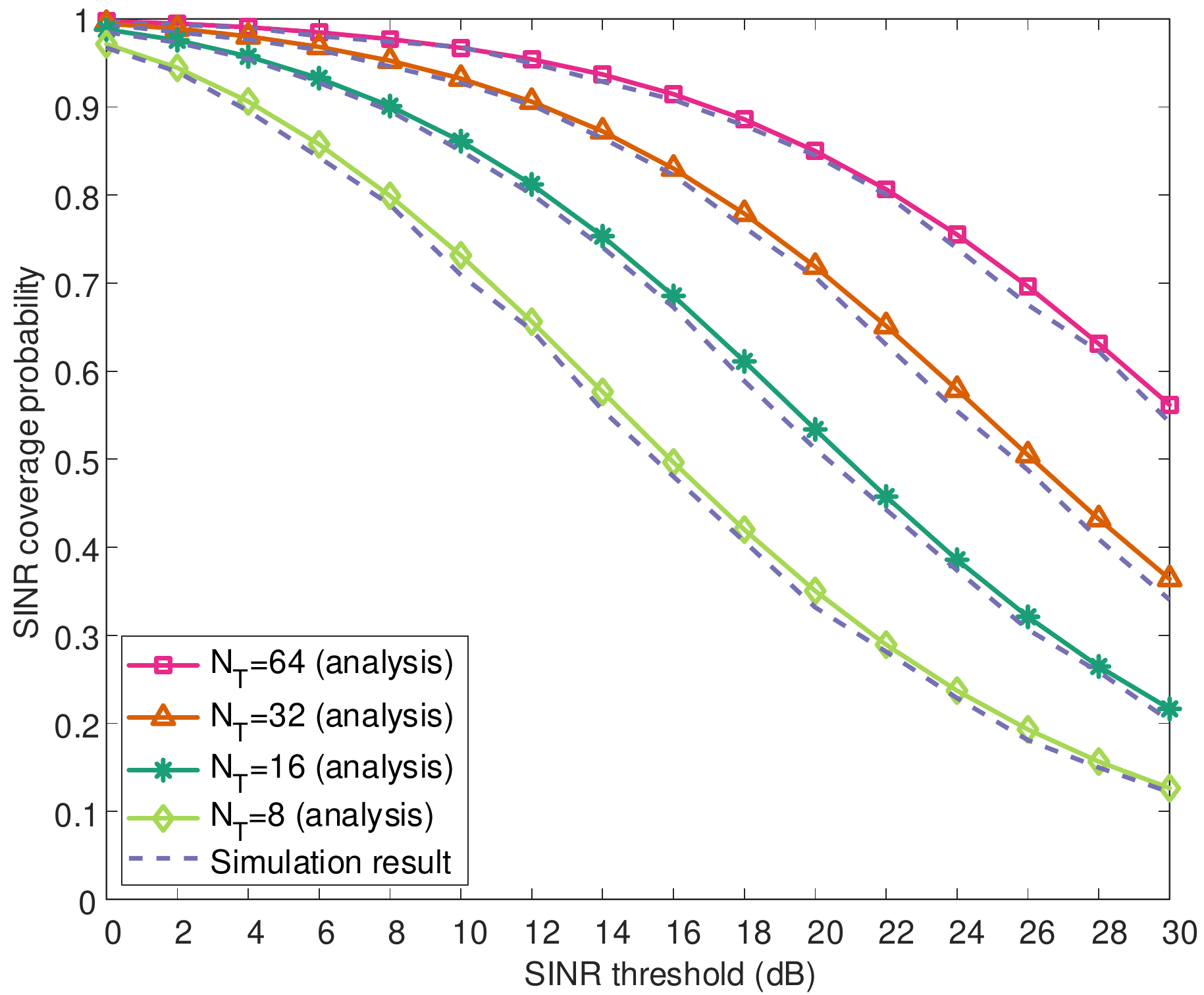}
  \label{THz_pcov_Nt}
  }
  \enspace
  \subfigure[]
  {\includegraphics[width=0.485\linewidth]{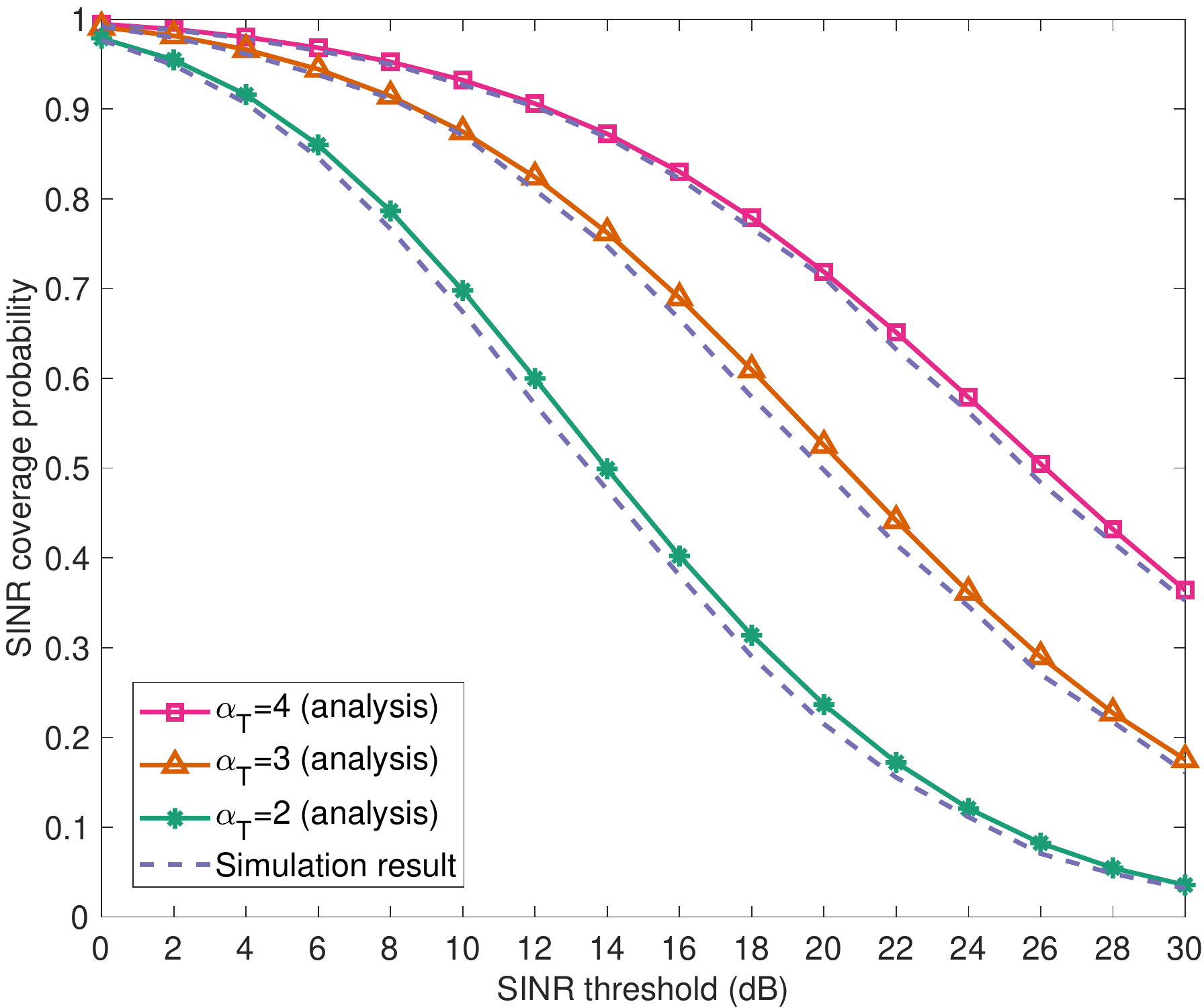}
  \label{THz_pcov_ple}
  }
\caption{$\tt SINR$ coverage probability of THz networks across various (a) antenna array size and (b) pathloss exponent.}
\label{pcov_Nt_ple}
\end{figure*}

\begin{figure*}[!htp]
  \centering
  \subfigure[]
  {\includegraphics[width=0.485\linewidth]{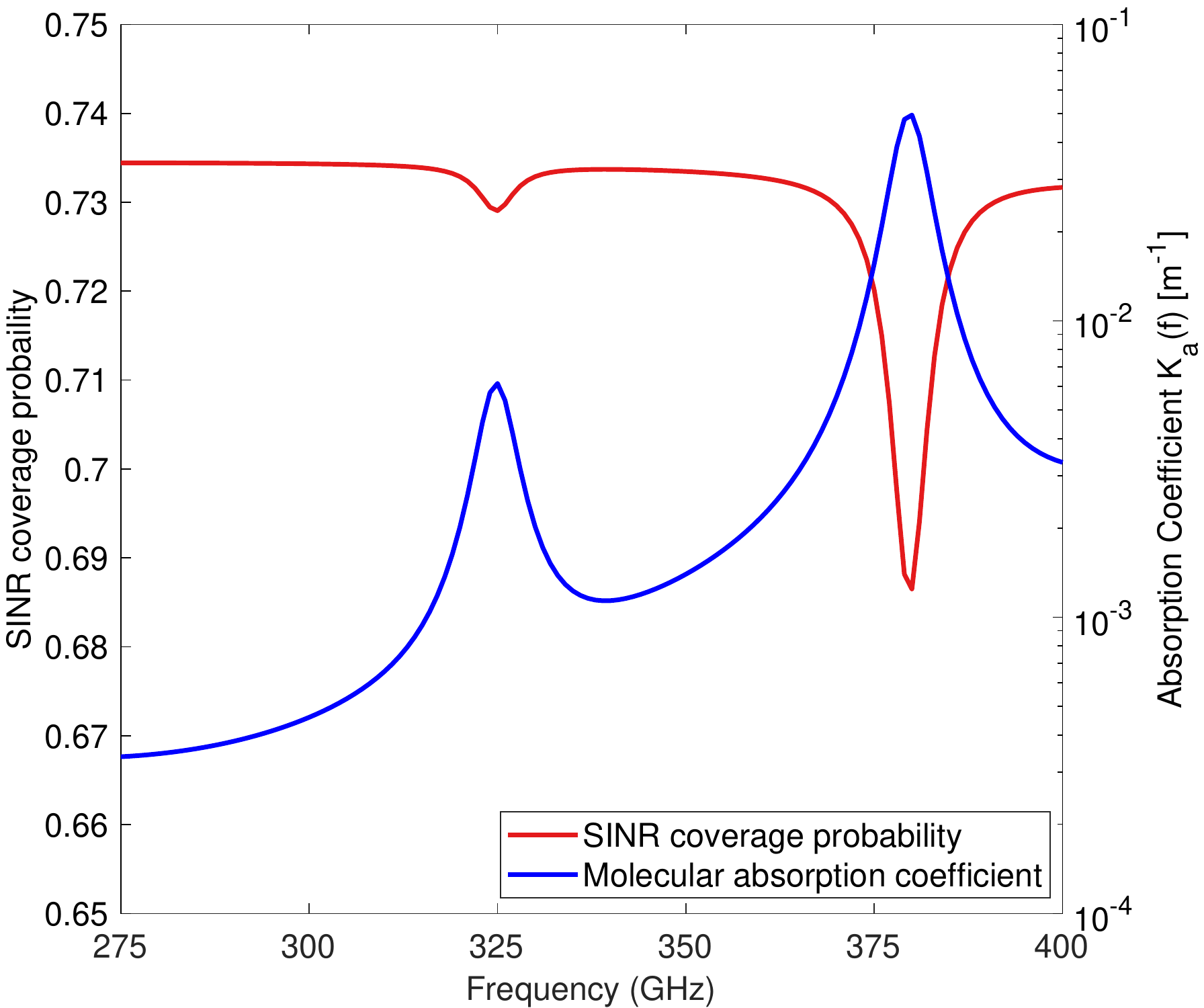}
  \label{THz_pcov_kaf}
  }
  \enspace
  \subfigure[]
  {\includegraphics[width=0.485\linewidth]{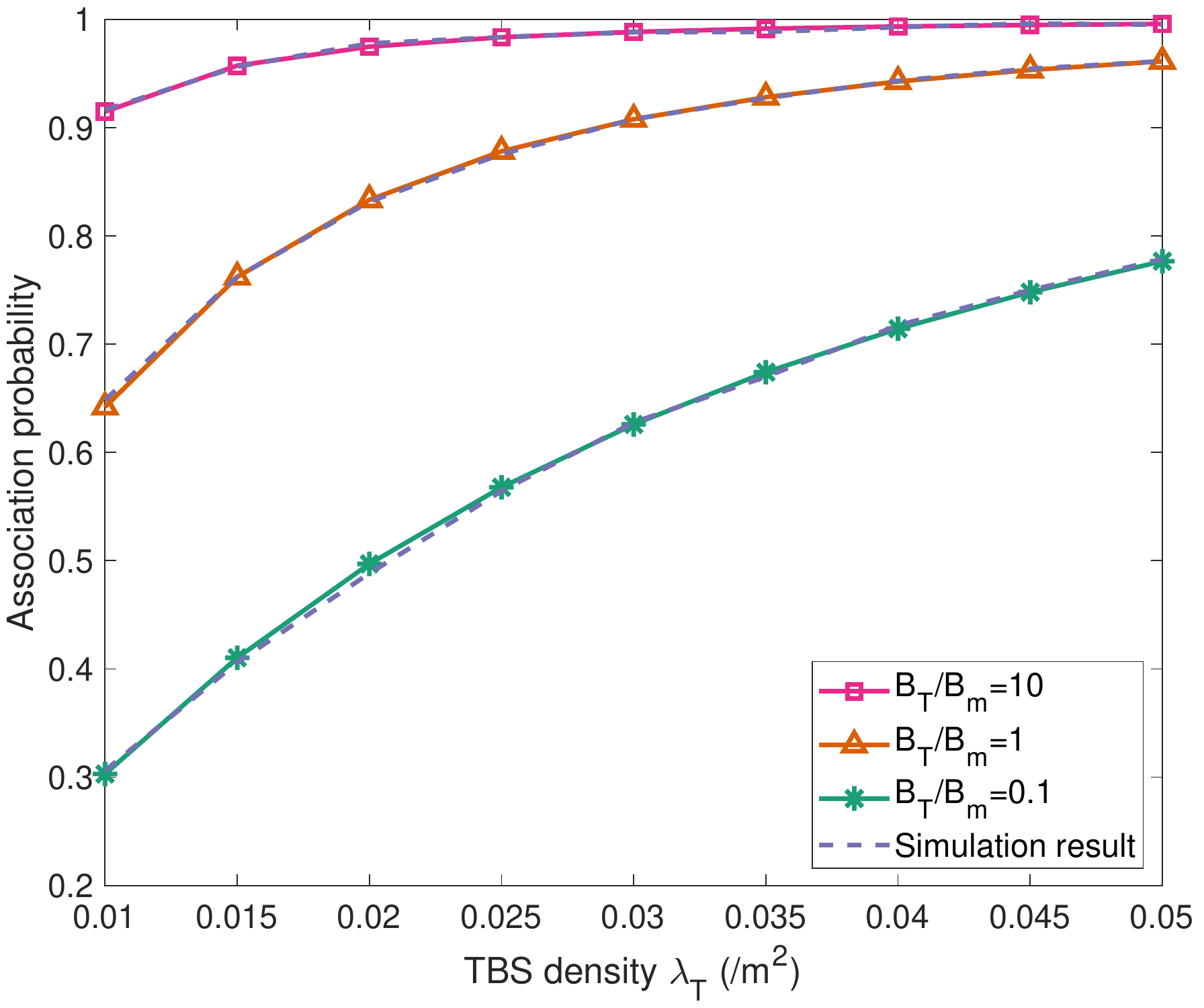}
  \label{THz_at_lambda}
  }
  
\caption{(a) Coverage probability of THz networks for varying operating frequency with $\tt SINR$ threshold of 20 dB, (b) Associate probability to TBS for various $\lambda_T$ and bias value.}
\label{THz_kf_at}
\end{figure*}

\begin{figure*}[!htp]
  \centering
  \subfigure[]
  {\includegraphics[width=0.485\linewidth]{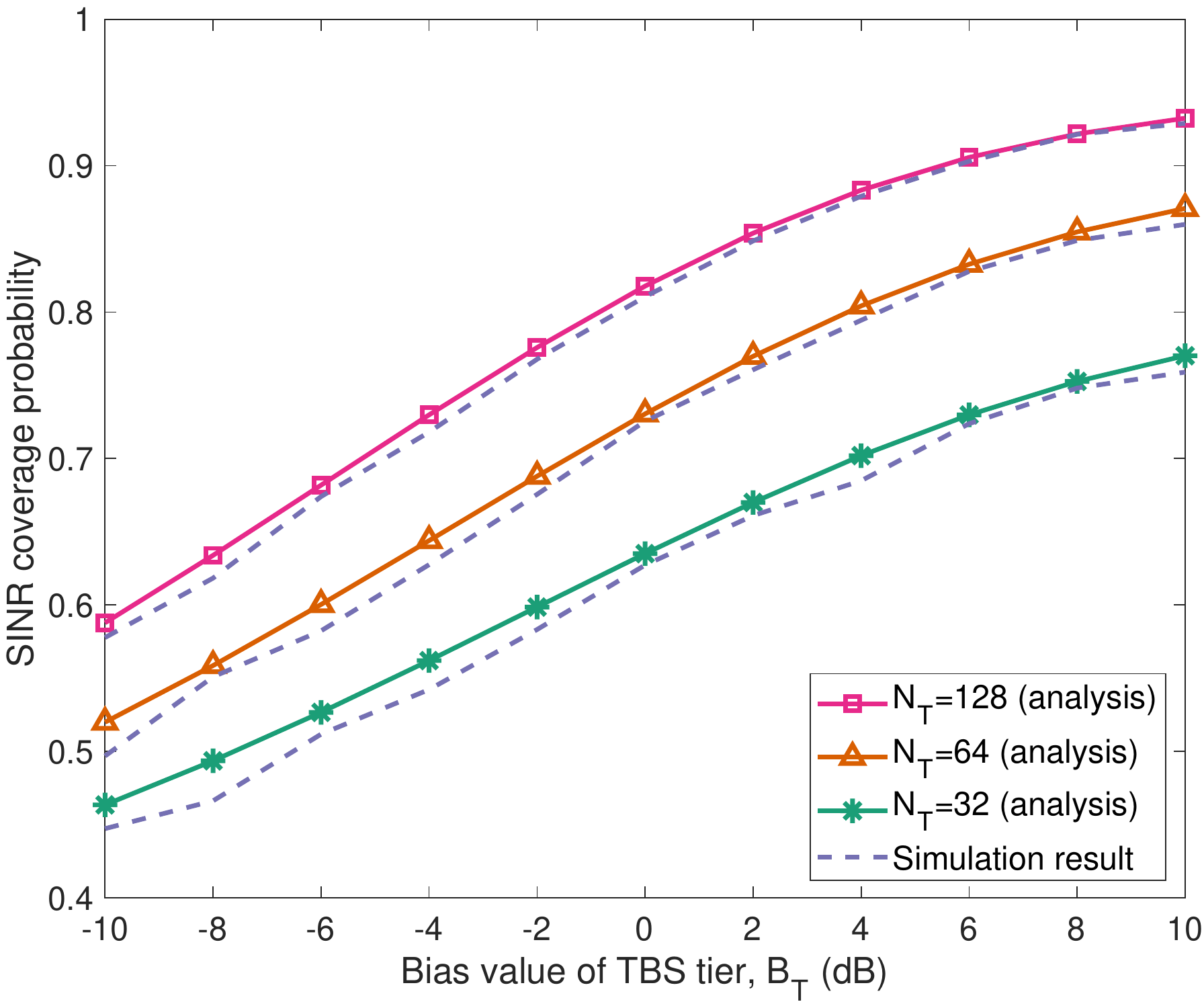}
  \label{p_cov_nt}
  }
  \enspace
  \subfigure[]
  {\includegraphics[width=0.485\linewidth]{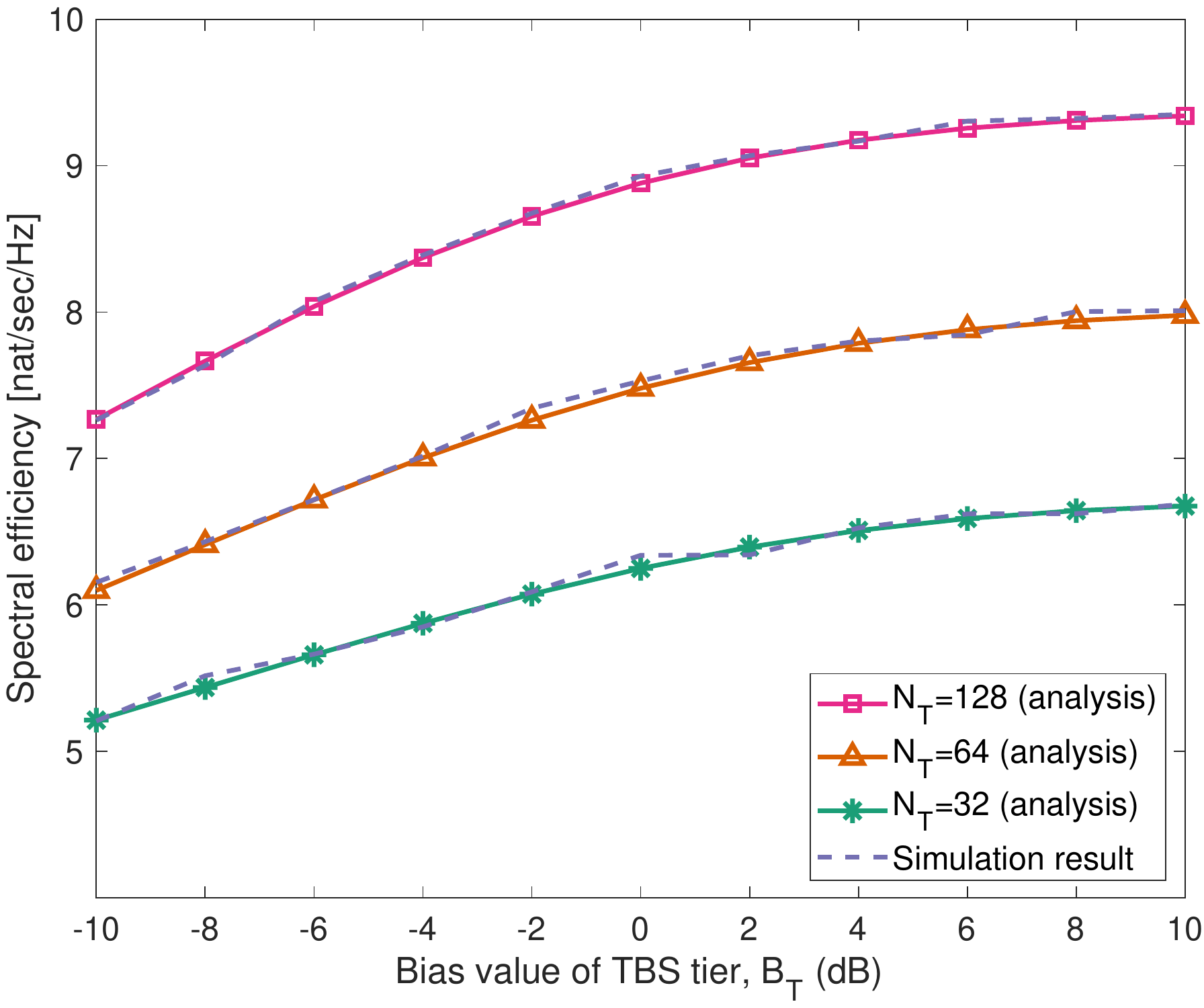}
  \label{se_vs_nt}
  }
  
\caption{(a) Coverage probability of the hybrid IoT network with TBS threshold 30 dB, $\lambda_T = 0.01/m^2$, (b) Spectral efficiency of the hybrid IoT network for various antenna array sizes.}
\label{p_cov_se}
\end{figure*}

\begin{figure*}[htp]
  \centering
  \subfigure[]
  {\includegraphics[width=0.485\linewidth]{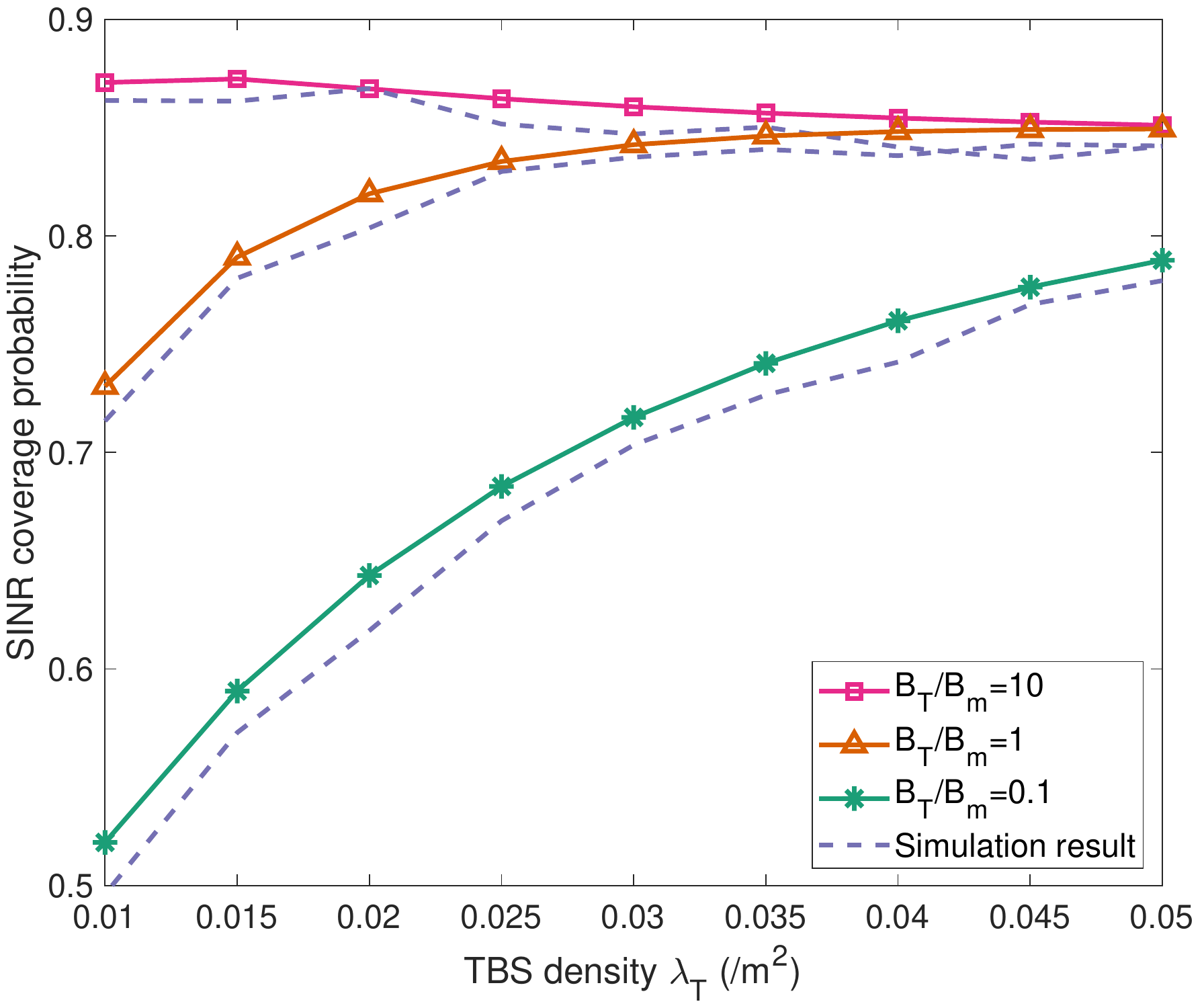}
  \label{p_cov_lam}
  }
  \enspace
  \subfigure[]
  {\includegraphics[width=0.485\linewidth]{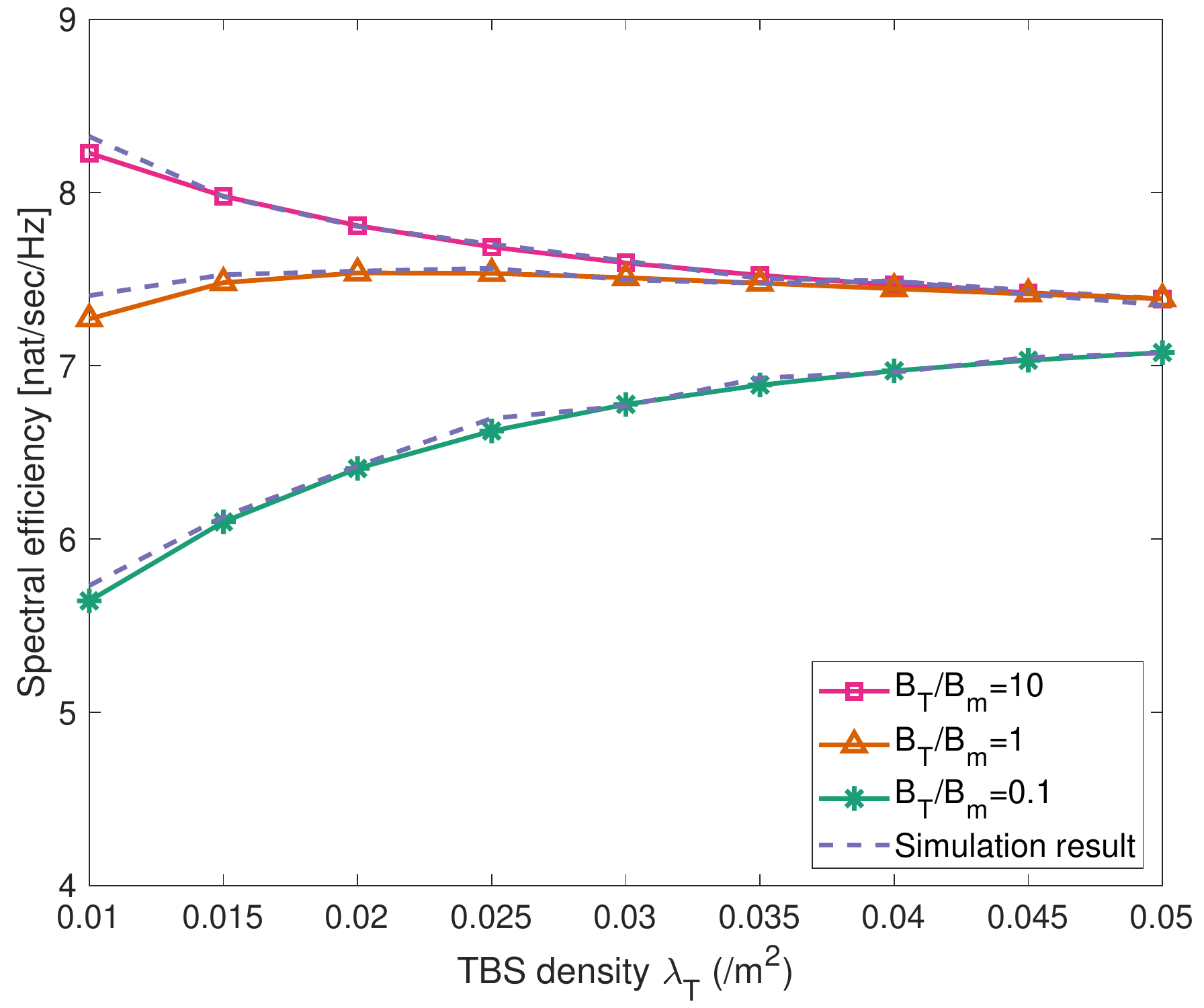}
  \label{se_vs_lam}
  }
  
\caption{(a) Coverage probability of the hybrid IoT network with threshold of 20 dB, (b) Spectral efficiency of the hybrid IoT network for varying TBS densities.}
\label{p_cov_se_lam}
\end{figure*}

\subsection{Hybrid IoT networks Performance}
The previous subsection validated the accuracy of the mathematical framework for THz networks. In this subsection, we will further verify the analytical results of the hybrid IoT network. Moreover, the impact of bias value, network density, and antenna array size on the overall performance of the hybrid IoT network will be investigated in this subsection. 

\subsubsection{Association Probability to THz Tier}
In Fig. \ref{THz_at_lambda}, we demonstrate the association probability with THz tier for a varying density of THz networks and a bias value. It is observed that the analytical results calculated by Theorem \ref{thm:associate_THz} achieve a decent accuracy comparing to the simulation results. According to the results, a larger TBS density will monotonically increase the association probability to the THz tier. Intuitively, larger TBS density will considerably extend the coverage leading to a higher association probability with the THz tier. Similarly, a higher bias value will also greatly encourage users to associate with the THz tier.

\subsubsection{Coverage Probability and Spectral Efficiency for Antenna Array Size and bias}
In Fig. \ref{p_cov_nt}, we illustrate the $\tt SINR$ coverage probability and spectral efficiency of the hybrid IoT network for varying bias value $B_T$ with $N_T=128,64,32$, respectively. It can be observed that there is a slight gap between the analytical and simulated coverage probability while the analytical and simulated spectral efficiency shows an accurate agreement. This is due to the Alzer's inequality used for coverage probability evaluation, which provides a tight bound. In contrast, the analytical spectral efficiency is evaluated by the exact expression with the method of Lemma \ref{lemm:useful technique}. Overall speaking, a larger TBS antenna array size and higher bias value to the THz tier significantly benefit both coverage probability and spectral efficiency. Especially, a larger antenna array of TBS will make up for its inherent poor penetrability and thus capitalize upon its abundant spectrum resources. Furthermore, a similar impact of the antenna array on spectral efficiency is also captured. Therefore, TBSs are better equipped with a larger antenna array in the hybrid IoT network if given a choice. Higher bias value to THz tier, which attracts more UE to associate with TBSs, also improves the network performance. In contrast, the effect of bias value in $\tt SINR$ coverage probability and spectral efficiency is different. Specifically, desired $\tt SINR$ coverage probability can be achieved flexibly by adopting a larger antenna array or higher bias value, whereas the benefit from higher bias to spectral efficiency is relatively limited. Higher target spectral efficiency can only be achieved by adopting a larger antenna array size. Therefore, these parameters need to be tuned accordingly depending on the network requirement and application scenarios.

\subsubsection{Coverage Probability and Spectral Efficiency for BS density and bias}
Fig. \ref{p_cov_se_lam} illustrates how the TBS density contributes to the coverage probability and spectral efficiency of the hybrid IoT network for varying TBS bias values. As depicted in the figure, a higher bias value can improve both the coverage probability and spectral efficiency given fixed TBS density. However, higher TBS density may degrade the network performance. In specific, the beneficial effect of a high TBS bias value is significantly evident in a low TBS density regime, whereas the benefit is gradually limited with increasing TBS density when UE is more encouraged to be associated with TBS, i.e., $B_T/B_m \geq 1$. High TBS density with a strong TBS bias may even lead to lower network performance. In contrast, when the TBS bias value is relatively low, e.g., $B_T/B_m < 1$, high TBS density can continuously improve the network performance. Therefore, bias and density can be tuned accordingly when designing the deployment strategy to optimize the overall performance. For example, when TBS is deployed on a small scale with a relatively low density, a higher TBS bias value is regarded as the priority factor in order to significantly boost the network performance. When TBSs are densely deployed, bias can only have a minimal impact on the network performance. This is because when TBSs are deployed with very high density, THz links will dominate the network regardless of the bias value.

\section{Conclusion}
In this paper, we investigated the THz networks and the hybrid IoT network consisting of TBSs and MBSs. We considered an accurate MLFT antenna pattern to analyze the THz network based on the stochastic geometric framework. The closed-form expression of the Laplace transform of the interference in THz networks was derived. Furthermore, the coverage probability and spectral efficiency of THz and the hybrid IoT network were evaluated. The numerical results revealed that both a larger antenna array and density of TBS improve the hybrid IoT network's performance. Furthermore, we recognized a fundamental trade-off relation between the TBS's node density and the bias to mmWave/THz. This work can be extended by using a non-homogeneous point process like the Poisson cluster process. Also, the blockage effect can be modeled by a more sophisticated and realistic blockage model proposed in \cite{bai2014analysis}.

\appendices
\section{Proof of Lemma \ref{lemm:serving_BS_distance_distribution}}
\label{appendix:proof_distance_distribution}
Considering the typical user is associated with THz tier, e.g., $k=$ T, the distribution of the distance between the associated TBS to the typical UE can be derived as below.
\begin{align}
\begin{split}
    &\mathbb{P}[X_T > \hat{x}]
    =\mathbb{P}[x_T > \hat{x} | k = T]\\
    &=\frac{\mathbb{P}[x_T > \hat{x}, k = T]}{\mathbb{P}[k = T]}\\
    &=\frac{\mathbb{P}[x_T > \hat{x}, k = T]}{\mathcal {A}_{\mathrm T}},
\end{split}    
\end{align}
where $x_T$ denotes the distance of the closest TBS to the typical UE and the joint probability in the numerator is
\begin{align}
\begin{split}
    &\mathbb{P}[x_T > \hat{x}, k = T]\\
    &=\mathbb{P}[x_T > \hat{x}, P_{r,T}>P_{r,M} ]\\
    &=\int_{\hat{x}}^{R_T}\mathbb{P}[P_{r,T}>P_{r,M}]f_{x_T}(x){\rm d}x\\
    &= \int_{\hat{x}}^{R_T} f_{x_T}(x) e^{-\pi \lambda_m \left(\varepsilon r^{\alpha_T} e^{k_a(f_T)x}\right)^ {\frac{2}{\alpha_m}}}{\rm d}x,
\end{split}    
\end{align}
where $\varepsilon$ and $f_{x_T}(x)$ are defined in Theorem \ref{thm:associate_THz}. Therefore, the PDF of $X_T$ is expressed as
\begin{align}
\begin{split}
    &f_{X_T}(\hat{x})=\frac{{\rm d}\left[1-\mathbb{P}[X_T > \hat{x}]\right]}{{\rm d} \hat{x}}\\
    &=\frac{1}{\mathcal {A}_{\mathrm T}} \left[f_{x_T}(\hat{x}) e^{-\pi \lambda_m \left(\varepsilon \hat{x}^{\alpha_T} e^{k_a(f_T)\hat{x}}\right)^ {\frac{2}{\alpha_m}}}\right].
    \label{eqn:X_T_distribution}
\end{split}    
\end{align}
Similarly, we can first calculate the joint probability as below.
\begin{align}
\begin{split}
    &\mathbb{P}[x_m > \hat{x}, k = m]\\
    &=\mathbb{P}[x_m > \hat{x}, P_{r,M}>P_{r,T} ]\\
    &=\int_{\hat{x}}^{R_m}\mathbb{P}[\varepsilon x_T^{\alpha_T} e^{k_a(f_T)x_T} > x^{\alpha_m}]f_{x_m}(x){\rm d}x\\
    &=\int_{\hat{x}}^{R_m}\mathbb{P}[r>\nu(x)]f_{x_m}(x){\rm d}x,\\
    &=\int_{\hat{x}}^{R_m}e^{-\pi\lambda_T \nu^2(x)}f_{x_m}(x){\rm d}x
\end{split}    
\end{align}
where $f_{x_m}(x) = \frac{2\pi\lambda_m}{1-e^{-\lambda_m \pi R_m^2}}  x e^{-\pi\lambda_m x^2}$, and  $\nu(x)=\frac{\alpha_T}{k_a(f_T)}W\left(\frac{k_a(f_T)}{\alpha_T}[\frac{x^{\alpha_m}}{\varepsilon}]^{\frac{1}{\alpha_T}}\right)$, where $W(x)$ is Lambert $W$ function. Then we get the PDF of $X_m$ as below.

\begin{align}
\begin{split}
    &f_{X_m}(\hat{x})=\frac{{\rm d}\left[1-\mathbb{P}[X_m > \hat{x}]\right]}{{\rm d} \hat{x}}\\
    &=\frac{1}{\mathcal {A}_{\mathrm m}} [f_{x_m}(\hat{x})e^{-\pi\lambda_T \nu^2(\hat{x})}].
    \label{eqn:X_m_distribution}
\end{split}    
\end{align}

\section{Proof of Theorem \ref{lemm:coverage_probability}}
\label{appendix:proof_lemma_coverage_probability}
Provided the derivation in (\ref{eqn:THz_coverage_1}), the coverage probability of TBS-only networks can be approximated as
\begin{align}
  \begin{split}
    \mathbb {C}_T\left(\tau\right) {\simeq} \sum_{n=1}^{M} \binom{M}{n}\left( -1\right)^{n+1} \mathbb{E}_{\hat{J}, x_T}\left[e^{-P_n(x_T) \left(\hat{J} + \hat{N}\right)}
    \right],
  \end{split}
\end{align}
And the term $\mathbb{E}_{\hat{J}, x_T}\left[e^{-P_n(x_T) \left(\hat{J} + \hat{N}\right)}\right]$ can be calculated as:
\begin{equation}
  \begin{split}
    &\mathbb{E}_{\hat{J}, x_T}\left[e^{-P_n(x_T)\left(\hat{J} + \hat{N}\right)} \right]\\
    &=\int_{0}^{R_T} f_{x_T}(x) \mathbb{E}_{\hat{J}}\left[e^{-P_n(x)\left(\hat{J} + \hat{N}\right)} \right] {\rm d}x\\
    &=
    \int_{0}^{R_T} f_{x_T}(x) e^{-P_n(x) \hat{N}}\mathcal{L}_{\hat{J}}\left( P_n(x)\right) {\rm d}x,
  \end{split}
  \label{THz_coverage_3}
\end{equation}
where $\mathcal{L}_{A}\left( s\right) = \mathbb{E}\left[e^{-s A} \right]$ is the Laplace transform of a parameter $A$, and given $x_T = x$, it can be evaluated as follows.
\begin{align}
    \mathcal{L}_{\hat{J}}(s) &= \mathbb{E} \left[e^{-s \sum_{i\in \Phi_{T \backslash o}}\hat{G_i}x^{-\alpha_T}_{T,i} g_i}\right] \notag\\
    &= \mathbb{E}\left[ \prod \limits_{i\in \Phi_{T \backslash o}} e^{-s\hat{G_i}x^{-\alpha_T}_{T,i} g_i}\right] \notag\\
    &\overset{(a)}{=} \mathbb{E}\left[ \prod \limits_{i\in \Phi_{T \backslash o}} \frac{1}{{(1 + \frac{s\hat{G_i}x^{-\alpha_T}_{T,i}}{M})}^M}\right]  \notag\\
    &\overset{(b)}{=}\mathbb{E}_{\Phi_T}\! \prod \limits_{i\in \Phi_{T \backslash o}}\! \left[ (1-2K_T \psi_T) {+} \sum_{k=1}^{K_T}\frac{2\psi_T}{{(1 + \frac{s\hat{G_k}x^{-\alpha_T}_{T,i}}{M})}^M} \right]\notag\\
    &\overset{(c)}{=}exp\left(4\pi\lambda_T\psi_T \left[\chi_T(s)+\frac{K_T}{2}\left(x^2-R_T^2\right)\right] \right)
  \label{THz_laplace_2}
\end{align}
where (a) is calculated by the moment generating function of Gamma random variable, (b) is by computing the mean of antenna gain, (c) follows the probability generating function (PGFL) of the PPP, $K_T=\bigl \lfloor \frac {N_T}{2}\bigr \rfloor$, $\psi_T$ is the the  half-power beamwidth and $\chi_T(s)$ is evaluated as follows:
\begin{align}
  \begin{split}
  & \chi_T(s) = \int_x^{R_T}\sum_{k=1}^{K_T}\frac{1}{{(1 + \frac{s\hat{G_k}t^{-\alpha_T}}{M})}^M} t{\rm d}t\\
  & =\sum_{k=1}^{K_T} \left[ \frac{t^2}{2} {}~_2F_1(-\frac{2}{\alpha_T},M;\frac{\alpha_T-2}{\alpha_T};-\frac{s\hat{G_k}t^{-\alpha_T}}{M}) \right]_x^{R_m},
  \end{split}
  \label{THz_laplace_3}
\end{align}
where ${}_2F_1(\cdot )$ denotes the hypergeometric function, and $[F(x)]_a^b=F(b)-F(a)$. 

\ifCLASSOPTIONcaptionsoff
  \newpage
\fi

\bibliographystyle{ieeetr}
\bibliography{ref}

%

\end{document}